\documentclass[11pt]{amsart}%
\usepackage{amsmath}
\usepackage{amsfonts}
\usepackage{amssymb}
\usepackage{graphicx}%
\providecommand{\U}[1]{\protect\rule{.1in}{.1in}}
\newtheorem{theorem}{Theorem}

\newtheorem{corollary}[theorem]{Corollary}

\newtheorem{lemma}[theorem]{Lemma}

\newtheorem{proposition}[theorem]{Proposition}
\newtheorem{remark}[theorem]{Remark}

\newcommand{\nil}[1]{}
\email{chams@aub.edu.lb}
\email{alain@connes.org}
\begin{document}
\begin{titlepage}
\vspace{.3cm} \vspace{1cm}
\begin{center}
\baselineskip=16pt \centerline{\large\bf  The Uncanny Precision of the Spectral Action } \vspace{2truecm} \centerline{\large\bf Ali H.
Chamseddine$^{1,3}$\ , \ Alain Connes$^{2,3,4}$\ \ } \vspace{.5truecm}
\emph{\centerline{$^{1}$Physics Department, American University of Beirut, Lebanon}}
\emph{\centerline{$^{2}$College de France, 3 rue Ulm, F75005, Paris, France}}
\emph{\centerline{$^{3}$I.H.E.S. F-91440 Bures-sur-Yvette, France}}
\emph{\centerline{$^{4}$Department of Mathematics, Vanderbilt University, Nashville, TN 37240 USA}}
\end{center}
\vspace{2cm}
\begin{center}
{\bf Abstract}
\end{center}
Noncommutative geometry has been slowly emerging as a new paradigm of geometry which starts from quantum mechanics. One of its key features is that the new geometry is spectral in agreement with the physical way of measuring distances. In this paper we present  a detailed introduction with an overview on the study of the quantum nature of space-time using the tools of noncommutative geometry. In particular we examine the suitability of using the {\it spectral action} as action functional for the theory. To demonstrate how the spectral action encodes the dynamics of gravity we examine the accuracy of the approximation of the spectral action by its asymptotic expansion in the case of the round sphere $S^3 $. We  find that the two terms corresponding to the cosmological constant and the scalar curvature term already give the full result with remarkable accuracy.   This is then applied to the physically relevant case of $S^3\times S^1 $ where we show that the spectral action in this case is also given, for any test function, by the sum of two terms  up to an astronomically small correction, and in particular all higher order terms $a_{2n}$ vanish. This result is confirmed by evaluating the spectral action using the heat kernel expansion where we check that the higher order terms $a_4$ and $a_6$ both vanish due to remarkable cancelations. We also show that the Higgs potential appears as an exact perturbation  when the test  function used is a smooth cutoff function.
\end{titlepage}

\section{ An overview}

Our experimental information on the nature of space-time is based on two sources:

\begin{itemize}
\item High energy physics  based on
cosmic ray information and particle accelerator experiments, whose results are encapsulated in the Standard Model
of particle physics.

\item Cosmology based on astronomical observations.
\end{itemize}

The large scale global picture is well described in terms of Riemannian
geometry and general relativity, but this picture breaks down at high energy  where the quantum
effects take over. It is thus natural to look for a paradigm of geometry
which starts from the quantum framework, where the role of real variables is
played by self-adjoint operators in Hilbert space. Such a framework for
geometry has been slowly emerging under the name of noncommutative geometry.
One of its key features, besides the ability to handle spaces for which
coordinates no longer commute with each other, is that this new geometry is
\emph{spectral}. This is in agreement with physics in which most of the data
we have, either about the far distant parts of the universe or about high
energy physics, are also of spectral nature. The red shifted spectra of
distant galaxies or the momentum eigenstates of outgoing particles in high
energy experiments both point towards a prevalence of spectral information. In the
same vein the existing unit of time (length) is also of spectral nature. From
the mathematical standpoint it takes some doing to obtain a purely spectral
(Hilbert space theoretical) counterpart of Riemannian geometry. One reason for
the difficulty of this task is that, as is well known since the examples of J.
Milnor \cite{Milnor}, non-isometric Riemannian spaces exist which have the
same spectra (for the Dirac or Laplacian operators). Another reason is that
the conditions for a (compact) space to be a smooth manifold are given in
terms of the local charts, whose existence and compatibility is assumed, but
whose intrinsic meaning is more elusive.

The paradigm of noncommutative geometry is that of spectral triple. As its
name indicates it is of spectral nature. By definition a spectral triple is a
unitary Hilbert space representation of \textquotedblleft something". This
something is an equipment that allows one to manipulate algebraically
coordinates and to measure distances. The algebra of the coordinates is
denoted by ${\mathcal{A}}$ and is an involutive algebra, with involution
$a\mapsto a^{\ast}$. The equipment needed to measure distances is the inverse
line element $D$ which is unbounded and fulfills $D=D^{\ast}$. Altogether
these data fulfill some algebraic relations, \textit{e.g.\/} if we talk about
the simplest geometric space \textit{i.e.\/}\ the circle $S^{1}$ the relation
between the complex unitary coordinate $U$ and the inverse line element $D$ is
just $[D,U]=U$, which is in the vein of the Heisenberg commutation relations.

Thus, a geometry is given as a Hilbert space representation of the pair
$({\mathcal{A}},D)$ and can be encoded by the spectral triple $({\mathcal{A}%
},{\mathcal{H}},D)$ where ${\mathcal{H}}$ is the Hilbert space in which both
the algebra ${\mathcal{A}}$ and the inverse line element $D$ are now
concretely represented, the latter as an unbounded self-adjoint operator. This
picture shares with the Wigner paradigm for a particle as an (irreducible)
representation of the Poincar\'e group the feature that it separates the
kinematical relations from the choice of the Hilbert space representation. It
is only when the latter is chosen that actual measurements of distances
between points $x$ and $y$ can be performed by formulas such as
\[
\mathrm{Distance}\,(x,y)=\sup\,|f(x)-f(y)|\,, \ f\in{\mathcal{A}}\,,
\ \Vert[D,f]\Vert\leq1
\]
where indeed the norm $\Vert[D,f]\Vert$ is the operator norm in Hilbert space
and depends on the specific choice of the representation.

We now have at our disposal a reconstruction theorem (\textit{cf\/}
\cite{CoRec}) which shows that ordinary Riemannian spaces are neatly
characterized among spectral triples by the following kinematical relations:

\begin{itemize}
\item The algebra ${\mathcal{A}}$ is commutative.

\item The commutator $[[D,a],b]=0$ for any $a,b\in{\mathcal{A}}$.

\item The following ``Heisenberg type" relation\footnote{Here the multiple
commutator is defined as
\[
[T_{1},T_{2},\ldots,T_{n}]=\sum_{\sigma}\epsilon(\sigma)\,T_{\sigma
(1)}T_{\sigma(2)}\cdots T_{\sigma(n)}
\]
} holds\footnote{We assume for simplicity that the dimension $n$ is odd} , for
some $a_{j}^{\alpha}\in{\mathcal{A}}$:
\begin{equation}
\label{orientability}\sum_{\alpha}a_{0}^{\alpha}[[D,a_{1}^{\alpha}%
],[D,a_{2}^{\alpha}],\ldots, [D,a_{n}^{\alpha}]]=1
\end{equation}

\end{itemize}

together with the following spectral requirements:

\begin{itemize}
\item The $k$-th characteristic value of the resolvent of $D$ is $O(k^{-1/n})$.

\item Regularity.

\item Absolute continuity.
\end{itemize}

We refer to \cite{CoRec} for the precise statement. The meaning of
\eqref{orientability} is that the determinant of the metric $g^{\mu\nu}$ does
not vanish, and more precisely that its square root multiplied by the volume
form $\sum_{\alpha}a_{0}^{\alpha}da_{1}^{\alpha}\wedge da_{2}^{\alpha}%
\wedge\cdots da_{n}^{\alpha}$ gives $1$. The reason for the last two spectral
requirements is technical and allows one to specify the regularity
($C^{\infty}$, real analytic...) of the space and to control the spectral
measures. The first of the spectral requirements is crucial in that it bounds
the \textquotedblleft effective dimension" of the spectrum of the space in the
representation. There are good physics reasons to consider that the apparent
dimension, equal to four, of space-time is governed by the asymptotic behavior
of the eigenvalues of the line element, which is the Euclidean propagator.
Moreover this spectral dimension is not restricted to be an integer a priori
and can model fractal dimension easily. The above reconstruction Theorem shows
furthermore that the operator $D$ in the spectral triple is a Dirac type
operator, \textit{i.e.\/}\ an order one operator with symbol given by a
representation of the Clifford algebra. The restriction to spin manifolds is
obtained by requiring a \emph{real structure} \textit{i.e.\/}\ an antilinear
unitary operator $J$ acting in ${\mathcal{H}}$ which plays the same role and
has the same algebraic properties as the charge conjugation operator in
physics. When the dimension $n$ involved in the reconstruction Theorem is even
(rather than odd) the right hand side of \eqref{orientability} is now replaced
by the chirality operator $\gamma$ which is just a ${\mathbb{Z}}/2$-grading in
mathematical terms. It fulfills the rules
\begin{equation}
\gamma^{2}=1\,,\ \ [\gamma,a]=0,\qquad a\in\mathcal{A} \label{chiral}%
\end{equation}
The following further relations hold for $D,J$ and $\gamma$
\begin{equation}
J^{2}=\varepsilon\,,\ DJ=\varepsilon^{\prime}JD,\quad J\,\gamma=\varepsilon
^{\prime\prime}\gamma J,\quad D\gamma=-\gamma D \label{eight}%
\end{equation}
where {$\varepsilon,\varepsilon^{\prime},\varepsilon^{\prime\prime}\in\left\{
-1,1\right\}  $. The values of the three signs $\varepsilon,\varepsilon
^{\prime},\varepsilon^{\prime\prime}$ depend only, in the classical case of
spin manifolds, upon the value of the dimension $n$ modulo $8$ and are given
in the following table \cite{CoSM}:\vspace{0.08in} }

\begin{center}
{%
\begin{tabular}
[c]{|c|rrrrrrrr|}\hline
\textbf{n } & 0 & 1 & 2 & 3 & 4 & 5 & 6 & 7\\\hline\hline
$\varepsilon$ & 1 & 1 & -1 & -1 & -1 & -1 & 1 & 1\\
$\varepsilon^{\prime}$ & 1 & -1 & 1 & 1 & 1 & -1 & 1 & 1\\
$\varepsilon^{\prime\prime}$ & 1 &  & -1 &  & 1 &  & -1 & \\\hline
\end{tabular}
}
\end{center}

\vspace{0.09in} In the classical case of spin manifolds there is thus a
relation between the metric (or spectral) dimension given by the rate of
growth of the spectrum of $D$ and the integer modulo $8$ which appears in the
above table. For more general spaces however the two notions of dimension (the
dimension modulo $8$ is called the $KO$-dimension because of its origin in
$K$-theory) become independent since there are spaces $F$ of metric dimension
$0$ but of arbitrary $KO$-dimension. More precisely, starting with an ordinary
spin geometry $M$ of dimension $n$ and taking the product $M\times F$, one
obtains a space whose metric dimension is still $n$ but whose $KO$-dimension
is the sum of $n$ with the $KO$-dimension of $F$, which as explained can take
any value modulo $8$. Thus, one now has the freedom to shift the
$KO$-dimension at very little expense \textit{i.e.\/}\ in a way which does not
alter the plain metric dimension. As it turns out the Standard Model with
neutrino mixing favors the shift of dimension from the $4$ of our familiar
space-time picture to $10=4+6=2$ modulo $8$ \cite{Barrett}, \cite{AC}. The
shift from $4$ to $10$ is a recurrent idea in string theory compactifications,
where the $6$ is the dimension of the Calabi-Yau manifold used to
\textquotedblleft compactify". Effectively the dimension $10$ is related to
the existence of Majorana-Weyl fermions. The difference between this approach
and ours is that, in the string compactifications, the metric dimension of the
full space-time is now $10$ which can only be reconciled with what we
experience by requiring that the Calabi-Yau fiber remains unnaturally small.
In order to learn how to perform the above shift of dimension using a
$0$-dimensional space $F$, it is important to classify such spaces. This was
done in \cite{beggar}, \cite{SM}. There, we classified the \emph{finite}
spaces $F$ of given $KO$-dimension. A space $F$ is finite when the algebra
${\mathcal{A}}_{F}$ of coordinates on $F$ is finite dimensional. We no longer
require that this algebra is commutative. The first key advantage of dropping
the commutativity can be seen in the simplest case where the finite space $F$
is given by
\begin{equation}
{\mathcal{A}}=M_{k}({\mathbb{C}})\,,\ {\mathcal{H}}=M_{k}({\mathbb{C}%
})\,,\ D=0\,,\ J\,\xi=\xi^{\ast},\qquad\xi\in{\mathcal{H}}_{F} \label{einsym}%
\end{equation}
where the algebra ${\mathcal{A}}=M_{k}({\mathbb{C}})$ is acting by left
multiplication in ${\mathcal{H}}=M_{k}({\mathbb{C}})$. We have shown in
\cite{cc2} that the study of pure gravity on the space $M\times F$ yields
Einstein gravity on $M$ minimally coupled with Yang-Mills theory for the gauge
group $\mathrm{SU}(k)$. The Yang-Mills gauge potential appears as the inner
part of the metric, in the same way as the group of gauge transformations (for
the gauge group $\mathrm{SU}(k)$) appears as the group of inner
diffeomorphisms. One can see in this Einstein-Yang-Mills example that the
finite geometry fulfills a nice substitute of commutativity (of ${\mathcal{A}%
}$) namely
\begin{equation}
\lbrack a,b^{0}]=0\,,\quad\forall\,a,b\in{\mathcal{A}}\, \label{orderzero}%
\end{equation}
where for any operator $a$ in ${\mathcal{H}}$, $a^{0}=Ja^{\ast}J^{\,-1}$. This
is called the order zero condition. Moreover the representation of
${\mathcal{A}}$ and $J$ in {${\mathcal{H}}$ is irreducible. This example is
(taking $\gamma=1$) of $KO$-dimension equal to $0$. In \cite{beggar} we
classified the irreducible $({\mathcal{A}},{\mathcal{H}},J)$ and found out
that the solutions fall into two classes. Let $\mathcal{A}_{\mathbb{C}}$ be
the complex linear space generated by $\mathcal{A}$ in ${\mathcal{L}%
}({\mathcal{H}})$, the algebra of operators in ${\mathcal{H}}.$ By
construction ${\mathcal{A}}_{\mathbb{C}}$ is a complex algebra and one only
has two cases: }

\begin{enumerate}
\item The center $Z\left(  \mathcal{A}_{\mathbb{C}}\right)  $ is $\mathbb{C}$,
in which case ${\mathcal{A}}_{\mathbb{C}}=M_{k}({\mathbb{C}})$ for some $k$.

\item The center $Z\left(  \mathcal{A}_{\mathbb{C}}\right)  $ is
$\mathbb{C\oplus C}$ and ${\mathcal{A}}_{\mathbb{C}}=M_{k}({\mathbb{C}})\oplus
M_{k}({\mathbb{C}})$ for some $k$.
\end{enumerate}

Moreover the knowledge of ${\mathcal{A}}_{\mathbb{C}}=M_{k}({\mathbb{C}})$
shows that ${\mathcal{A}}$ is either $M_{k}({\mathbb{C}})$ (unitary case),
$M_{k}({\mathbb{R}})$ (real case) or, when $k=2\ell$ is even, $M_{\ell
}({\mathbb{H}})$, where ${\mathbb{H}}$ is the field of quaternions (symplectic
case). This first case is a minor variant of the Einstein-Yang-Mills case
described above. It turns out by studying their ${\mathbb{Z}}/2$ gradings
$\gamma$, that these cases are incompatible with $KO$-dimension $6$ which is
only possible in case (2). If one assumes that one is in the \textquotedblleft
symplectic--unitary" case and that the grading is given by a grading of the
vector space over ${\mathbb{H}}$, one can show that the dimension of
${\mathcal{H}}$ which is $2k^{2}$ in case (2) is at least $2\times16$ while
the simplest solution is given by the algebra ${\mathcal{A}}=M_{2}%
({\mathbb{H}})\oplus M_{4}({\mathbb{C}})$. This is an important variant of the
Einstein-Yang-Mills case because, as the center $Z\left(  \mathcal{A}%
_{\mathbb{C}}\right)  $ is $\mathbb{C\oplus C}$, the product of this finite
geometry $F$ by a manifold $M$ appears, from the commutative standpoint, as
two distinct copies of $M$. We showed in \cite{beggar} that requiring that
these two copies of $M$ stay a finite distance apart reduces the symmetries
from the group $\mathrm{SU}(2)\times\mathrm{SU}(2)\times\mathrm{SU}(4)$ of
inner automorphisms\footnote{of the even part of the algebra} to the
symmetries $U(1)\times\mathrm{SU}(2)\times\mathrm{SU}(3)$ of the Standard
Model. This reduction of the gauge symmetry occurs because of the second
kinematical condition $[[D,a],b]=0$ which in the general case becomes:
\begin{equation}
\lbrack\lbrack D,a],b^{0}]=0\,,\quad\forall\,a,b\in{\mathcal{A}}\,
\label{orderone}%
\end{equation}
Thus the noncommutative space singles out $4^{2}=16$ as the number of physical
fermions, the symmetries of the standard model emerge, and moreover, as shown
in \cite{mc2}, the model predicts the existence of right-handed neutrinos, as
well as the see-saw mechanism. In the above Einstein-Yang-Mills case, the
Yang-Mills fields appeared as the inner part of the metric in the same way as
the group of gauge transformations (for the gauge group $\mathrm{SU}(k)$)
appeared as the group of inner diffeomorphisms. But in that case all fields
remained massless. It is the existence of a non-zero $D$ for the finite space
$F$ that generates the Higgs fields and the masses of the Fermions and the $W$
and $Z$ fields through the Higgs mechanism. The new fields are computed from
the kinematics but the action functional, the \emph{spectral action}, uses in
a crucial manner the representation in Hilbert space. In order to explain the
conceptual meaning of this spectral action functional it is important to
understand in which way it encodes gravity in the commutative case. As
explained above the spectrum of the Dirac operator (or similarly of the
Laplacian) does not suffice to encode an ordinary Riemannian geometry. However
the Einstein-Hilbert action functional, given by the integral of the scalar
curvature multiplied by the volume form, appears from the heat expansion of
the Dirac operator. More generally it appears as the coefficient of
$\Lambda^{2}$ in the asymptotic expansion for large $\Lambda$ of the trace
\begin{equation}
\mathrm{Tr}(f(D/\Lambda))\sim2\Lambda^{4}f_{4}a_{0}+2\Lambda^{2}f_{2}%
a_{2}+f_{0}a_{4}+\ldots+\Lambda^{-2k}f_{-2k}a_{4+2k}+\ldots\label{heat}%
\end{equation}%
%TCIMACRO{\FRAME{ftbpFU}{4.2134in}{2.7544in}{0pt}{\Qcb{Cutoff function f}%
%}{\Qlb{figure1}}{ff6.eps}{\special{ language "Scientific Word";
%type "GRAPHIC";  maintain-aspect-ratio TRUE;  display "USEDEF";
%valid_file "F";  width 4.2134in;  height 2.7544in;  depth 0pt;
%original-width 4.1632in;  original-height 2.7121in;  cropleft "0";
%croptop "1";  cropright "1";  cropbottom "0";
%filename 'ff6.eps';file-properties "XNPEU";}} }%
%BeginExpansion
\begin{figure}
[ptb]
\begin{center}
\includegraphics[
height=2.7544in,
width=4.2134in
]%
{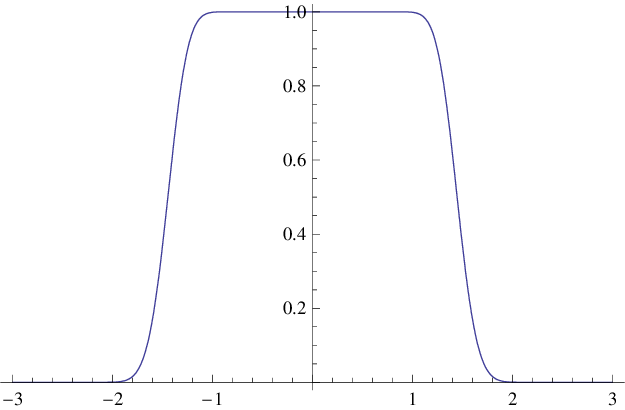}%
\caption{Cutoff function f}%
\label{figure1}%
\end{center}
\end{figure}
%EndExpansion
when the Riemannian geometry $M$ is of dimension $4$, and where $f$ is a
smooth even function with fast decay at infinity. The choice of the function
$f$ only enters in the multiplicative factors $f_{4}=\int_{0}^{\infty
}f(u)u^{3}du$, $f_{2}=\int_{0}^{\infty}f(u)udu$, $f_{0}=f(0)$ and
$f_{-2k}=\left(  -1\right)  ^{k}\frac{k!}{\left(  2k\right)  !}f^{(2k)}(0)$,
\textit{i.e.\/}\ the derivatives of even order at $0$, for $k\geq0$. Thus,
when $f$ is a \textquotedblleft cutoff" function (\textit{cf\/} Figure
\ref{figure1}) it has vanishing Taylor expansion at $0$ and the asymptotic
expansion \eqref{heat} only has three terms:
\begin{equation}
\mathrm{Tr}(f(D/\Lambda))\sim2\Lambda^{4}f_{4}a_{0}+2\Lambda^{2}f_{2}%
a_{2}+f(0)a_{4} \label{heat1}%
\end{equation}
The term in $\Lambda^{4}$ is a cosmological term, the term in $\Lambda^{2}$ is
the Einstein-Hilbert action functional, and the constant term $a_{4}$ gives
the integral over $M$ of curvature invariants such as the square of the Weyl
curvature and topological terms such as the Gauss-Bonnet, with numerical
coefficients of order one. It is thus natural to take the expression
$\mathrm{Tr}(f(D/\Lambda))$ as a natural spectral formulation of gravity. We
are working in the Euclidean formulation \textit{i.e.\/}\ with a signature
$(+,+,+,+)$ and the Euclidean space-time manifold is taken to be compact for
simplicity. In the non-compact case we have shown in \cite{dilaton} how to
replace the simple counting of eigenvalues of $|D|$ of size $<\Lambda$
given\footnote{for $f$ a cutoff function} by \eqref{heat}, by a localized
counting. This simply introduces a dilaton field. We also tested this idea of
taking the expression $\mathrm{Tr}(f(D/\Lambda))$ as a natural spectral
formulation of gravity by computing this expression in the case of manifolds
with boundary and we found \cite{boundary} that it reproduces exactly the
Hawking-Gibbons \cite{GH} additional boundary terms which they introduced in
order to restore consistency and obtain Einstein equations as the equations of
motion in the case of manifolds with boundary. Further, Ashtekar et al
\cite{ashtekar} have recently shown that the use of the Dirac operator in a
first order formalism, which is natural in the noncommutative setting, avoids
the tuning and subtraction of a constant term. One may be worried by the large
cosmological term $\Lambda^{4}f_{4}a_{4}$ that appears in the spectral action.
It is large because the value of the cutoff scale $\Lambda$ is dictated,
roughly speaking, by the Planck scale since the term $\Lambda^{2}f_{2}a_{2}$
is the gravitational action $\frac{1}{16\pi G}\int R\sqrt{g}d^{4}x$. Thus it
seems at first sight that the huge cosmological term $\Lambda^{4}f_{4}a_{4}$
overrides the more subtle Einstein term $\Lambda^{2}f_{2}a_{2}$. There is,
however, and even at the classical level to which the present discussion
applies a simple manner to overcome this difficulty. Indeed the kinematical
relation \eqref{orientability} in fact fixes the Riemannian volume form to
be\footnote{up to a numerical factor}
\begin{equation}
\sqrt{g}d^{4}x=\sum_{\alpha}a_{0}^{\alpha}da_{1}^{\alpha}\wedge da_{2}%
^{\alpha}\wedge da_{3}^{\alpha}\wedge da_{4}^{\alpha} \label{vform}%
\end{equation}
Thus, if we vary the metric with this constraint we are in the context of
unimodular gravity \cite{uni}, and the cosmological term cancels out in the
computation of the conditional probability of a gravitational configuration
with total volume $V$ held fixed. The remaining unknown, then, is the
distribution of volumes $d\mu(V)$, which is just a distribution on the
half-line ${\mathbb{R}}_{+}\ni V$. The striking conceptual advantages of the
spectral action are

\begin{itemize}
\item Simplicity: when $f$ is a cutoff function, the spectral action is just
\emph{counting} the number of eigenstates of $D$ in the range $[-\Lambda
,\Lambda]$.

\item Positivity: when $f\geq0$ (which is the case for a cutoff function) the
action $\mathrm{Tr}(f(D/\Lambda))\geq0$ has the correct sign for a Euclidean action.

\item Invariance: one is used to the diffeomorphism invariance of the
gravitational action but the functional $\mathrm{Tr}(f(D/\Lambda))$ has a much
stronger invariance group, the unitary group of the Hilbert space
${\mathcal{H}}$.
\end{itemize}

One price to pay is that, as such, the action functional $\mathrm{Tr}%
(f(D/\Lambda))$ is not local. It only becomes so when it is replaced by the
asymptotic expansion \eqref{heat1}. This suggests that one should at least
compute the next term in the asymptotic expansion (even though this term
appears multiplied by the second derivative $f^{\prime\prime}(0)=0$ when $f$
is a cutoff function) just to get some idea of the size of the remainder. In
fact both $D$ and $\Lambda$ have the physical dimension of a mass, and there
is no absolute scale on which they can be measured. The ratio $D/\Lambda$ is
dimensionless and the dimensionless number that governs the quality of the
approximation \eqref{heat1} can be chosen to just be the number $N(\Lambda)$
of eigenvalues $\lambda$ of $D$ whose size is less than $\Lambda$,
\textit{i.e.\/}\ $|\lambda|\leq\Lambda$. When $f$ is a cutoff function the
size of the error term in \eqref{heat1} should be $O(N^{-k})$ for any positive
$k$, using the flatness of the Taylor expansion of $f$ at $0$. In the case of
interest, where $M$ is the Euclidean space-time, a rough estimate of the size
of $N$ is the $4$-dimensional volume of $M$ in Planck units \textit{i.e.\/}%
\ an order of magnitude\footnote{using the age of the universe in Planck units
to estimate the spatial Euclidean directions and the inverse temperature
$\beta=1/kT$ also in Planck units, to set the size of the imaginary time
component of the Euclidean $M$.} of $N\sim10^{214}$(at the present radius, see
section two for details). Thus, even without the vanishing of $f^{\prime
\prime}(0)$, the rough error term $N^{-1/2}\sim10^{-107}$ is quite small in the
approximation of the spectral action by its local version \eqref{heat1}. We shall
in fact show that a much better estimate holds in the simplified model of Euclidean
 space-time given by the product $S_a^3\times S^1_\beta$. Another advantage of the
above spectral description of the gravitational action is that one can now use
the same action $\mathrm{Tr}(f(D/\Lambda))$ for spaces which are not
Riemannian. The simplest case is the product of a Riemannian geometry $M$ (of
dimension $4$) by the finite space $F$ of \eqref{einsym}. The only new term
that appears is the Yang-Mills action functional of the $\mathrm{SU}(k)$ gauge
fields which form the inner part of the metric. This new term appears as an
additional term in the coefficient $a_{4}$ of $\Lambda^{0}$, and with the
positive sign. In other words gravity on the slightly noncommutative space
$M\times F$ gives ordinary gravity minimally coupled with $\mathrm{SU}%
(k)$-Yang-Mills gauge theory. The latter theory is massless and the fermions
are in the adjoint representation. The fermionic part of the action is easy to
write since one has the operator $D$ whose inner fluctuations are
\begin{equation}
D_{A}=D+A+JAJ^{-1}\,,\ \ A=\sum a_{j}[D,b_{j}]\,,\ a_{j},b_{j}\in{\mathcal{A}%
}\,,\ A=A^{\ast}\label{innerfluct}%
\end{equation}
In the Einstein-Yang-Mills system so obtained, all fields involved are massless.

\medskip We now consider the product $M\times F$ of a Riemannian geometry $M$
(of dimension $4$) by the finite space $F$ of $KO$-dimension $6$ which was
determined above. The computation shows that (\textit{cf\/} \cite{mc2})

\begin{itemize}
\item The inner fluctuations of the metric give an $U(1)\times\mathrm{SU}%
(2)\times\mathrm{SU}(3)$ gauge field and a complex Higgs doublet scalar field.

\item The spectral action $\mathrm{Tr}(f(D/\Lambda))$ plus the antisymmetric
bilinear form $\langle J\xi,D_{A}\eta\rangle$ on chiral fermions, gives the
Standard Model minimally coupled to gravity, with the Majorana mass terms and
see-saw mechanism.

\item The gauge couplings fulfill the unification constraint, the Yukawa
couplings fulfill $Y_{2}=4g^{2}$, where $Y_{2}$ is defined in
eq(\ref{yequation}), and the Higgs quartic coupling also fulfills a
unification constraint.
\end{itemize}

Most of the new terms occur in the $a_{4}$ term of the expansion
\eqref{heat1}. This is the case for the minimal coupling of the Higgs field as
well as its quartic self-interaction. The terms $a_{0}$ and $a_{2}$ get new
contributions from the Majorana masses (\textit{cf\/} \cite{mc2}), but the
main new term in $a_{2}$ has the form of a mass term for the Higgs field with
the coefficient $-\Lambda^{2}$. This immediately raises the question of the
meaning of the specific values of the couplings in the above action
functional. Unlike the above massless Einstein-Yang-Mills system we can no
longer take the above action simply as a classical action, would it be because
of the unification of the three gauge couplings, which does not hold at low
scale. The basic idea proposed in \cite{cc2} is to consider the above action
as an effective action valid at the unification scale $\Lambda$ and use the
Wilsonian approach of integrating the high frequency modes to show that one
obtains a realistic picture after \textquotedblleft running down" from the
unification scale to the energies at which observations are done. This
approach is closely related to the approach of Reuter \cite{reuter}, Dou and
Percacci \cite{Dou}, \cite{Perc}. The coarse graining uses a much lower scale
$\rho$ which can be understood physically as the resolution with which the
system is observed. The modes with momenta larger than $\rho$ cannot be
directly observed and their effect is averaged out by the functional integral.
In fact the way the renormalization group is computed in \cite{Dou} shows that
the derivative $\rho\partial_{\rho}\Gamma_{\rho}$ of the effective action is
expressed as a trace of an operator function of the propagators and is thus of
a similar nature as the spectral action itself, though the trace involves all
fields and not just the spin $\frac{1}{2}$ fields as in the spectral action.
It is an open question to compute the renormalization group flow for the
spectral action in the context of spectral triples. One expects, as explained
above, that new terms involving traces of functions of the bosonic
propagator\footnote{We thank John Iliopoulos for discussions on this point.}
$\frac{\delta^{2}}{\delta D\delta D}\mathrm{Tr}(f(D/\Lambda))$ will be
generated. The idea of taking the spectral action as a boundary condition of
the renormalization group at unification scale generates a number of severe
tests. The first ones involve the dimensionless couplings. These include

\begin{enumerate}
\item The three gauge couplings

\item The Yukawa couplings

\item The Higgs quartic coupling
\end{enumerate}

As is well known, the gauge couplings do not unify in the Standard Model but
the meeting of $g_{2}$ and $g_{3}$ specifies a \textquotedblleft unification"
scale of $\sim10^{17}$ GeV. For the Yukawa couplings the boundary condition
gives
\begin{equation}
Y_{2}=4\,g^{2},\ \ \ Y_{2}=\sum_{\sigma}\,(y_{\nu}^{\sigma})^{2}%
+(y_{e}^{\sigma})^{2}+3\,(y_{u}^{\sigma})^{2}+3\,(y_{d}^{\sigma})^{2}.
\label{yequation}%
\end{equation}
This yields a value of the top mass which is $1.04$ times the observed value
when neglecting\footnote{See \cite{mc2} for the precise satement.} the Yukawa couplings of the bottom quarks etc...and is hence
compatible with experiment. The Higgs quartic coupling (scattering parameter)
has the boundary condition of the form:
\[
\tilde{\lambda}(\Lambda)=\,g_{3}^{2}\,\frac{b}{a^{2}}\sim g_{3}^{2}%
\]
The numerical solution to the RG equations with the boundary value
$\lambda_{0}=0.356$ at $\Lambda=10^{17}$ GeV gives $\lambda(M_{Z})\sim0.241$
and a Higgs mass of the order of $170$ GeV. This value now seems to be ruled
out experimentally but this might simply be a clear indication of the presence
of some new physics, instead of the \textquotedblleft big desert" which is
assumed here in the huge range of energies between $10^{2}$ GeV and $10^{17}$
GeV. To be more precise the above \textquotedblleft prediction" of the Higgs
mass is in perfect agreement with the one of the Standard Model, when one
assumes the \textquotedblleft big desert" (\textit{cf\/} \cite{Hambey}). In a
forthcoming paper \cite{hidden} we show that the choice of the spectral
function $f$ could play an important role, even when it varies slightly from
the cutoff function. This is related to the fact that the vev of the Higgs
field is proportional to the scale $\Lambda$ and thus higher order corrections
do contribute. This will cause the relation between the gauge coupling
constants to be modified and to change the Higgs potential. Such gravitational
corrections are known to cause sizable changes to the Higgs mass
\cite{sturmia}.

The next tests involve the dimensionful couplings. These include

\begin{enumerate}
\item The inverse Newton constant $Z_{g}=1/G$.

\item The mass term of the Higgs.

\item The Majorana mass terms.

\item The cosmological constant.
\end{enumerate}

Since our action functional combines gravity and the Standard Model, the
analysis of \cite{Dou} applies, and the running of the couplings $Z$ which
have the physical dimension of the square of a mass is well approximated by
$\beta_{Z}=a_{1}k^{2}$ where the parameter $k$ is fixing the cutoff scale but
is considered itself as one of the couplings, while the coefficient $a_{1}$ is
a dimensionless number of order one. For the inverse $Z_{g}$ of the Newton
constant, one gets the solution:%

\begin{equation}
Z_{g}=\bar{Z}_{g}(1+\frac{1}{2}a_{1}\frac{k^{2}}{\bar{Z}_{g}}) \label{runzg1}%
\end{equation}
which behaves like a constant and shows that the change in $Z_{g}$ is moderate
between the low energy value $\bar{Z}_{g}$ at $k=0$ and its value at $k=m_{P}$
the Planck scale, for which $\frac{k^{2}}{\bar{Z}_{g}}=1$. We have shown in
\cite{mc2} that a relation between the moments of the cutoff function $f$
involved in the spectral action, of the form $f_{2}\sim5f_{0}$ suffices to
give a realistic value of the Newton constant, provided one applies the
spectral action at the unification scale $\Lambda\sim10^{17}$ GeV. The above
discussion of the running of $Z_{g}$ shows that this yields a reasonable low
energy value of the Newton constant $G$.

The form $\beta_{Z}=a_{1}k^{2}$ of the running of a coupling with mass$^{2}$
dimension implies that, as a rule, even if this coupling happens to be small
at low scale, it will necessarily be of the order of $\Lambda^{2}$ at
unification scale. For the Majorana mass terms, we explained in \cite{mc2} why
they are of the order of $\Lambda^{2}$ at unification and their role in the
see-saw mechanism shows that one should not expect them to be small at small
scale, thus a running like \eqref{runzg1} is realistic. Things are quite
different for the mass term of the Higgs. The spectral action delivers a huge
mass term of the form $-\Lambda^{2}H^{2}$ and one can check that it is
\emph{consistent} with the sign and order of magnitude of the quadratic
divergence of the self-energy of this scalar field. However though this shows
compatibility with a small low energy value it does by no means allow one to
justify such a small value. Giving the term $-\Lambda^{2}H^{2}$ at unification
scale and hoping to get a small value when running the theory down to low
energies by applying the renormalization group, one is facing a huge fine
tuning problem. Thus one should rather try to find a physical principle to
\emph{explain} why one obtains such a small value at low scale. In the
noncommutative geometry model $M\times F$ of space-time the size of the finite
space $F$ is governed by the inverse of the Higgs mass. Thus the above problem
has a simple geometric interpretation: \emph{Why is the space $F$ so
large\footnote{by a factor of $10^{16}$.} in Planck units?} There is a
striking similarity between this problem and the problem of the large size of
space in Planck units. This
suggests that it would be very worthwhile to develop  cosmology
in the context of the noncommutative geometry model of space-time, with in particular the
preliminary step of the Lorentzian formulation of the spectral action.

This also brings us to the important role played by the dilaton field which
determines the scale $\Lambda$ in the theory. The spectral action is taken to
be a function of the twisted Dirac operator so that $D^{2}$ is replaced with
$e^{-\phi}D^{2}e^{-\phi}.$ In \cite{dilaton} we have shown that the spectral
action is scale invariant, except for the dilaton kinetic energy. Moreover,
one can show that after rescaling the physical fields, the scalar potential of
the theory will be independent of the dilaton at the classical level. At the
quantum level, the dilaton acquires a Coleman-Weinberg potential \cite{CW} and
will have a vev of the order of the Planck mass \cite{Bubu}. The fact that the
Higgs mass is damped by a factor of $e^{-2\phi}$, can be the basis of an
explanation of the hierarchy problem.

In this paper we investigate the accuracy of the approximation of the spectral action
by the first terms of its asymptotic expansion. We consider the concrete example given by
the four-dimensional geometry $S_a^{3}\times S_\beta^{1}$ where $S_a^{3}$ is the
round sphere of radius $a$ as a model of space, while $S_\beta^{1}$ is
a circle of radius $\beta$ viewed as a model of imaginary periodic time at inverse
temperature $\beta$. We compute directly the spectral action and compare it with the
sum of the first terms of the asymptotic expansion. In section two we start with the round sphere $S_a^{3}$ and
use the known spectrum of the Dirac operator together with the Poisson summation formula, to estimate  the
remainder when using a smooth test function. This is then applied to the
four-dimensional space $S_a^{3}\times S_\beta^{1}$ where it is shown that, for natural
test functions, the spectral
action is completely determined by the first two terms, with an error of the order of
 $10^{-\sigma^2}$ where $\sigma$ is the inner diameter  $\Lambda\mu$, $\mu=\inf (a,\beta)$ in units
 of the cutoff $\Lambda$. Thus for instance an inner diameter of $10$ in cutoff units yields the accuracy of the
 first hundred decimal places, while an inner diameter of $10^{31}$ corresponding to the visible universe
 at inverse temperature of $3$ Kelvin and a cutoff at Planck scale\footnote{while the age of the
universe in Planck units gives $\Lambda a \sim 10^{61}$.}, yields an astronomical precision of $10^{62}$ accurate  decimal places.
 This is then extended in the presence of Higgs fields.
  The above direct computation allows one to double check coefficients in the spectral action. It also implies,
   for $S_a^{3}\times S_\beta^{1}$, the vanishing of all the Seeley-De Witt coefficients $a_{2n}$, $n\geq 2$, in the heat
  expansion of the square of the Dirac operator. This is confirmed in section three, by a local
   computation of the heat kernel expansion, where it is shown that $a_4$ and $a_6$ vanish due to subtle
   cancelations.

\bigskip

\section{Estimate of the asymptotics}

\label{estispecaction}

The number $N(\Lambda)$ of eigenvalues of $|D|$ which are $\leq\Lambda$
\begin{equation}
N(\Lambda) = \# \ \hbox{eigenvalues of $D$ in} \ [-\Lambda,\Lambda] ,
\end{equation}
is a step function $N(\Lambda)$ which jumps by the integer multiplicity of an eigenvalue
 whenever $\Lambda$ belongs to the spectrum of $|D|$. This integer valued function is the superposition of two terms,
\[
N(\Lambda)= \langle N(\Lambda) \rangle+N_{\mathrm{osc}} (\Lambda).
\]
The oscillatory part $N_{\mathrm{osc}} (\Lambda)$ is generically the same as for a random
matrix. The average part $\langle N(\Lambda) \rangle$ is
computed by a semiclassical approximation from local expressions involving the
familiar heat equation expansion and will now be carefully defined assuming
an expansion of the form\footnote{the $a_{\alpha}$
defined here is equal to the Seeley-de Witt coefficients $a_{n+2\alpha}$ in dimension $n$.}
\begin{equation}
\mathrm{Trace}\,(e^{-t\Delta})\sim\sum\,a_{\alpha}\,t^{\alpha}\qquad
(t\rightarrow0) \label{eq1.69}%
\end{equation}
for the positive operator $\Delta=D^{2}$. One has,
\begin{equation}
\Delta^{-s/2}=\frac{1}{\Gamma\left(  \frac{s}{2}\right)  }\int_{0}^{\infty
}e^{-t\Delta}\,t^{s/2-1}\,dt \label{eq1.68}%
\end{equation}
and the relation between the asymptotic expansion \eqref{eq1.69}
and the $\zeta$ function,
\begin{equation}
\zeta_{D}(s)=\mathrm{Trace}\,(\Delta^{-s/2}) \label{eq1.70}%
\end{equation}
is given by,

\begin{itemize}
\item $\alpha<0$ gives a \textit{pole} at $-2\alpha$ for $\zeta_{D}$ with
\begin{equation}
\mathrm{Res}_{s=-2\alpha}\,\zeta_{D}(s)=\frac{2\,a_{\alpha}}{\Gamma(-\alpha)}
\label{eq1.71}%
\end{equation}

\item $\alpha= 0$ (no $\log t$ term) gives regularity  at 0 for $\zeta_{D}$
with
\begin{equation}
\label{eq1.72}\zeta_{D} (0) = a_{0}\,.
\end{equation}

\end{itemize}

For simple superpositions of exponentials, as Laplace transforms,
\begin{equation}
f(u)=\int_{0}^{\infty}e^{-su}\,h(s)\,ds \label{eq1.73}%
\end{equation}
we can write formally,
\begin{equation}
f(t\Delta)=\int_{0}^{\infty}e^{-st\Delta}\,h(s)\,ds \label{eq1.74}%
\end{equation}
and
\begin{equation}
\mathrm{Trace}\,(f(t\Delta))\sim\sum\,a_{\alpha}\,t^{\alpha}\int_{0}^{\infty
}s^{\alpha}\,h(s)\,ds\,. \label{eq1.75}%
\end{equation}
For $\alpha<0$ one has,
\[
s^{\alpha}=\frac{1}{\Gamma(-\alpha)}\int_{0}^{\infty}e^{-sv}\,v^{-\alpha
-1}\,dv
\]
and
\[
\int_{0}^{\infty}s^{\alpha}\,h(s)\,ds=\frac{1}{\Gamma(-\alpha)}\int
_{0}^{\infty}f(v)\,v^{-\alpha-1}\,dv
\]
so that
\begin{align}
\mathrm{Trace}\,(f(t\Delta))  &  \sim\sum_{\alpha<0}\ \frac{1}{2}%
\ \mathrm{Res}_{s=-2\alpha}\,\zeta_{D}(s)\int_{0}^{\infty}f(v)\,v^{-\alpha
-1}\,dv\,t^{\alpha}\nonumber\\
&  +\zeta_{D}(0)\,f(0)+\sum_{\alpha>0}\,a_{\alpha}\,t^{\alpha}\int_{0}%
^{\infty}s^{\alpha}\,h(s)\,ds\,. \label{eq1.76}%
\end{align}
Now we assume that the only $\alpha>0$ for which $a_\alpha\neq 0$ are integers and
note that
\begin{equation}
\int_{0}^{\infty}s^{n}\,h(s)\,ds=(-1)^{n}\,f^{(n)}(0)\,, \label{eq1.77}%
\end{equation}
so that all the terms $a_\alpha$ for $\alpha>0$ have vanishing coefficients
 when $f$ is a cutoff function which is constant equal to $1$ in a
 neighborhood of $0$. To define the average part we consider the limit
  case $f(v)=1$ for $|v|\leq1$ and 0 elsewhere and get for the coefficients of \eqref{eq1.76}
\begin{equation}
\frac{1}{2}%
\ \int_{0}^{\infty}f(v)\,v^{-\alpha
-1}\,dv\,t^{\alpha}=\frac{t^\alpha}{(-2\alpha)}\,, \label{eq1.78}%
\end{equation}
which, with $t=\Lambda^{-2}$, gives the following \emph{definition} for the average part
\begin{equation}
\langle N(\Lambda)\rangle:=\sum_{k>0}\,\frac{\Lambda^{k}}{k}\,\mathrm{Res}%
_{s=k}\,\zeta_{D}(s)+\zeta_{D}(0)\,. \label{eq1.98}%
\end{equation}
To get familiar with this definition we shall work out its meaning in a simple case,

\begin{proposition}
\label{theorem1}Assume that \textrm{Spec }$D\subset\mathbb{Z}$ and that
 the total multiplicity of $\left\{  \pm n\right\}  $ is $P\left(  n\right)
$ for a polynomial $P(x)=\sum\,c_{k}\,x^{k}$. Then
one has
\[
\left\langle N\left(  \Lambda\right)  \right\rangle =%
{\displaystyle\int_{0}^{\Lambda}}
P\left(  u\right)  du+ c \,, \ \  c=\sum\;c_{k}\;\zeta({-k})\,,
\]
where $\zeta$ is the Riemann zeta function.
\end{proposition}

\begin{proof}
One has by construction, with $P(x)=\sum\,c_{k}\,x^{k}$,
\[
\zeta_{D}(s)=\sum\;P(n)\;n^{-s}=\,\sum\;c_{k}\;\zeta({s-k})
\]
 Thus
\[
\mathrm{Res}_{s=k}^{{}}\,\zeta_{D}(s)=\,c_{k-1}%
\]
and
\[
\langle N(\Lambda)\rangle:=\sum_{k>0}\,\frac{\Lambda^{k}}{k}\,\,c_{k-1}%
+\zeta_{D}(0)\,.
\]
The constant $\zeta_{D}(0)$ is given by
\[
\sum\;c_{k}\;\zeta({-k})
\]
and is independent of $\Lambda$.
\end{proof}

\subsection{The sphere $S^4$}
We check the hypothesis of Proposition \ref{theorem1} for a round even sphere.
We recall (\cite{trautman}) that the spectrum of the Dirac operator for the round sphere $S^{n}$ of unit
radius is given by
\begin{equation}
\mathrm{Spec}(D)=\{\pm(\frac{n}{2}+k)\;|\;k\in{\mathbb{Z}},k\geq 0\} \label{spD}%
\end{equation}
where the multiplicity of $(\frac{n}{2}+k)$ is equal to $2^{[\frac{n}{2}%
]}{\binom{k+n-1}{k}}$. Thus for $n=4$ one gets that the spectrum consists of
the relative integers, except for $\{-1,0,1\}$. The multiplicity of the
eigenvalue $m$ is $4{\binom{k+3}{k}}$ for $k+2=m$ which gives, for the total
multiplicity of $\pm m$
\[
P(m)=\frac{4}{3}(m+1)m(m-1)=\frac{4}{3}(m^{3}-m)
\]
which shows that one gets the correct minus sign for the scalar curvature term
after integration using Proposition \ref{theorem1}.  Thus one gets (up to the normalization
factor $\frac{4}{3}$ )
\begin{equation}
\mathrm{Tr}(|D|^{-s})=\zeta(s-3)-\zeta(s-1)\label{4sphere}%
\end{equation}
This function has a  value at $s=0$ given by
\[
\zeta(-3)-\zeta(-1)=\frac{1}{120}+\frac{1}{12}=\frac{11}{120}%
\]
which, taking into account the factor $\frac 43$ from normalization,
 matches
the coefficient $\frac{11}{360}\times 4$ which appears in the spectral action in front
of the Gauss-Bonnet term, as will be shown in  \S \ref{seeley}.

\subsection{The sphere $S^3$}

We now want to look at the
case of $S^{3}$ and determine how good the approximation of \eqref{eq1.76} is for test
functions.

In order to estimate the remainder of \eqref{eq1.76} in this special
case we shall use the Poisson summation formula
\begin{equation}
\sum_{\mathbb{Z}}h(n)=\sum_{\mathbb{Z}}\hat{h}(n)\,,\ \ \ \hat{h}%
(x)=\int_{\mathbb{R}}h(u)e^{-2\pi ixu}du \label{Poisson}%
\end{equation}
or rather, since the spectrum is $\frac{1}{2}+{\mathbb{Z}}$ in the odd case, the variant
\begin{equation}
\sum_{\mathbb{Z}}g(n+\frac{1}{2})=\sum_{\mathbb{Z}}(-1)^{n}\hat{g}(n)
\label{Poisson1}%
\end{equation}
 (obtained from \eqref{Poisson} using  $h(u)=g(u+\frac{1}{2})$).

 In the case of the three sphere, the
eigenvalues are $\pm(\frac{3}{2}+k)$, for $k\geq0$ with the multiplicity
$2{\binom{k+2}{k}}$. Thus $n+\frac{1}{2}$ has multiplicity $n(n+1)$. This
holds not only for $n\geq0$ but also for $n\in{\mathbb{Z}}$ since the
multiplicity of $-(n+\frac{1}{2})$ is $n(n+1)=m(m+1)$ for $m=-n-1$. In
particular $\pm\frac{1}{2}$ is not in the spectrum. Thus when we evaluate
$\mathrm{Tr}(f(D/\Lambda))$, with $f$ an even function, we get the following
sum
\begin{equation}
\mathrm{Tr}(f(D/\Lambda))=\sum_{\mathbb{Z}}n(n+1)f((n+\frac{1}{2})/\Lambda )
\label{specact}%
\end{equation}
 We apply \eqref{Poisson1} with $g(u)=(u^{2}-\frac{1}%
{4})f(u/\Lambda)$. The Fourier transform of $g$ is
\[
\hat{g}(x)=\int_{\mathbb{R}}g(u)e^{-2\pi ixu}du=\int_{\mathbb{R}}(u^{2}%
-\frac{1}{4})f(u/\Lambda )e^{-2\pi ixu}du
\]%
\[
=   \Lambda^{3}\int_{\mathbb{R}}v^{2}f(v)e^{-2\pi i\Lambda
xv}dv-\frac{1}{4} \Lambda   \int_{\mathbb{R}}f(v)e^{-2\pi
i\Lambda xv}dv
\]
We introduce the function $\hat{f}^{(2)}$ which is the Fourier transform of
$v^{2}f(v)$ and we thus get from \eqref{Poisson1},
\begin{equation}
\mathrm{Tr}(f(D/\Lambda))= \Lambda^{3}\sum_{\mathbb{Z}%
}(-1)^{n}\hat{f}^{(2)}(\Lambda n)-\frac{1}{4}  \Lambda
\sum_{\mathbb{Z}}(-1)^{n}\hat{f}(\Lambda n) \label{Poissonbis}%
\end{equation}
If we take the function $f$ in the Schwartz space ${\mathcal{S}}({\mathbb{R}%
})$, then both $\hat{f}$ and $\hat{f}^{(2)}$ have rapid decay and we can
estimate the sums
\[
 \sum_{n\neq0}|\hat{f}(\Lambda n)|\leq
C_{k}\Lambda^{-k} \,,\ \ \sum_{n\neq0}|\hat{f}^{(2)}(\Lambda  n)|\leq C_{k}\Lambda^{-k}%
\]
which gives, for any given $k$, an estimate for a sphere of radius $a$ of the form:
\begin{equation}
\mathrm{Tr}(f(D/\Lambda))=\left(  \Lambda a\right)  ^{3}\int_{\mathbb{R}}%
v^{2}f(v)dv-\frac{1}{4}\left(  \Lambda a\right)  \int_{\mathbb{R}%
}f(v)dv+O(\left(  \Lambda a\right)  ^{-k}) \label{Poissonfine}%
\end{equation}
The radius simply rescales $D$ and enters in such a way as
to make the product $\Lambda a$ dimensionless. This can be seen by noting that
the ratio $\frac{D}{\Lambda}$ contains the term $\frac{1}{\Lambda}e_{\alpha}^{\mu
}\gamma^{\alpha}\partial_{\mu}$ and the radius enters as $\frac{1}{a}$ in the
inverse dreibein $e_{\alpha}^{\mu}.$
Note that, provided that $k>1$
 one controls the constant in front of $(\Lambda a)^{-k}$ from the
constants $c_j$ with
\[
|x^k\hat{f}(x)|\leq
c_{1} \,,\ \   |x^k\hat{f}^{(2)}(x)|\leq c_{2} \,.
\]
To get an estimate of these constants $c_{j}$, say for $k=2$, one can use the
$L^{1}$-norms of the functions $\Delta f(v)$ and $\Delta ( v^{2}f(v))$
where $\Delta=-\partial_{v}^{2}$ is the Laplacian. If we take for $f$ a smooth
cutoff function we thus get that the $c_{j}$ are of order one.

In fact we shall soon get a much better estimate (Corollary \ref{cor1} below) which will show
that, for suitable test functions, a size of $N$ in cutoff units, $\Lambda a \sim N$, already
ensures a precision of the order of $e^{-N^2}$. We shall work directly with the physically more relevant model consisting of the product $S^3\times S^1$  viewed as a model of the imaginary time periodic compactification of space-time at a given temperature. Our estimates will work well for a size in cutoff units as small as $N\sim 10$ and will give the result with an astronomical precision for larger values. These correspond to later times since both the radius of space and the inverse temperature are increasing functions of time in this simple model.

\subsection{The product $S^3\times S^1$}

We now want to move to the $4$-dimensional Euclidean case obtained by taking
the product $M=S^{3}\times S^{1}$ of $S^{3}$ by a small circle. We take the
product geometry of a three dimensional geometry with Dirac operator $D_{3}$
by the one dimensional circle geometry with Dirac
\begin{equation}
\label{Done}D_{1}=\, \frac{1}{\beta}\,i\,\nabla_{\theta}%
\end{equation}
so that the spectrum of $D_{1}$ is $\frac1\beta({\mathbb{Z}}+\frac12)$.

\begin{lemma}
\label{theoerem2}Let $D$ be the Dirac operator of the product geometry
\begin{equation}
D=\left(
\begin{array}
[c]{cc}%
0 & D_{3}\otimes1+\,i\otimes D_{1}\\
D_{3}\otimes1-i\otimes D_{1} & 0
\end{array}
\right)  \label{prodD}%
\end{equation}
The asymptotic expansion for $\Lambda\rightarrow\infty$ of the spectral action
of $D$ is given by
\begin{equation}
\mathrm{Tr}(h(D^{2}/\Lambda^{2}))\sim\,2\,\beta\,\Lambda\,\mathrm{Tr}%
(k(D_{3}^{2}/\Lambda^{2}))\,, \label{specD}%
\end{equation}
where the function $k$ is given by
\begin{equation}
k(x)=\,\int_{x}^{\infty}\,(u-x)^{-1/2}\,h(u)\,du \label{transg}%
\end{equation}

\end{lemma}

\begin{proof}
By linearity of both sides in the function $h$ (using the linearity of the
transformation \eqref{transg}) it is enough to prove the result for the
function $h(x)=\,e^{-bx}$. One has
\[
D^{2}=\left(
\begin{array}
[c]{cc}%
D_{3}^{2}\otimes1+1\otimes D_{1}^{2} & 0\\
0 & D_{3}^{2}\otimes1+1\otimes D_{1}^{2}%
\end{array}
\right)
\]
and
\[
\mathrm{Tr}(e^{-b\,D^{2}/\Lambda^{2}})=\,2\,\mathrm{Tr}(e^{-b\,D_{1}%
^{2}/\Lambda^{2}})\,\mathrm{Tr}(e^{-b\,D_{3}^{2}/\Lambda^{2}})
\]
Moreover by \eqref{Done} the spectrum of $D_{1}$ is $\frac{1}{\beta
}({\mathbb{Z}}+\frac{1}{2})$ so that, using \eqref{Poisson1}, and for fixed $\beta$ and $b$,
one has for all $k>0$,
\[
\mathrm{Tr}(e^{-b\,D_{1}^{2}/\Lambda^{2}})\sim\sqrt{\pi}\,\beta\,\Lambda
\,b^{-1/2}+O(\Lambda^{-k})\,.
\]
Thus
\[
\mathrm{Tr}(e^{-b\,D^{2}/\Lambda^{2}})=\,2\,\beta\,\Lambda
\,\mathrm{Tr}(\sqrt{\pi}\,b^{-1/2}\,e^{-b\,D_{3}^{2}/\Lambda^{2}})+O(\Lambda^{-k+3})
\]
and the equality \eqref{specD} follows from
\[
\int_{x}^{\infty}\,(u-x)^{-1/2}\,e^{-bu}\,du=\,\sqrt{\pi}\,b^{-1/2}\,e^{-bx}%
\]
which shows that the function $k$ associated to $h(x)=e^{-bx}$ by the linear
transformation \eqref{transg} is $k(x)=\sqrt{\pi}\,b^{-1/2}\,e^{-bx}$.
\end{proof}

One can write \eqref{transg} in the form
\begin{equation}
k(x)=\,\int_{0}^{\infty}\,v^{-1/2}\,h(x+v)\,dv\,, \label{transg1}%
\end{equation}
which shows that $k$ has right support contained in the right support of $h$
{\textit{i.e.\/}\ that if $h$ vanishes identically on $[a,\infty\lbrack$ so
does $k$. It also gives a good estimate of the derivatives of $k$ since
\[
\partial_{x}^{n}k(x)=\,\int_{0}^{\infty}\,v^{-1/2}\,\partial_{x}%
^{n}h(x+v)\,dv\,.
\]
In fact, in order to estimate the size of the remainder in the asymptotic expansion of the spectral action for the product
$M=S^{3}\times S^{1}$, we shall now use the two dimensional form of
\eqref{Poisson1},
\begin{equation}
\sum_{{\mathbb{Z}}^{2}}g(n+\frac{1}{2},m+\frac{1}{2})=\sum_{{\mathbb{Z}}^{2}%
}(-1)^{n+m}\hat{g}(n,m) \label{Poisson2}%
\end{equation}
where the Fourier transform is given by
\begin{equation}
\hat{g}(x,y)=\int_{{\mathbb{R}}^{2}}g(u,v)e^{-2\pi i(xu+yv)}dudv
\label{fourier}%
\end{equation}
For the operator $D$ of \eqref{prodD}, and taking for $D_3$ the Dirac operator of the $3$-sphere $S_a^3$ of radius $a$, the eigenvalues of
$D^{2}/\Lambda^{2}$ are obtained by collecting the following
\[
(\frac{1}{2}+n)^{2}\left(  \Lambda a\right)  ^{-2}+(\frac{1}{2}+m)^{2}\left(
\Lambda\beta\right)  ^{-2}\,, \ \ n,m\in \mathbb Z%
\]
with the multiplicity $2n(n+1)$ for each $n,m\in \mathbb Z$. Thus, more precisely
\[
\mathrm{Tr}(h(D^{2}/\Lambda^{2}))=\sum_{{\mathbb{Z}}^{2}}2n(n+1)h((\frac
{1}{2}+n)^{2}\left(  \Lambda a\right)  ^{-2}+(\frac{1}{2}+m)^{2}\left(
\Lambda\beta\right)  ^{-2})
\]
which is of the form:
\begin{equation}
\mathrm{Tr}(h(D^{2}/\Lambda^{2}))=\sum_{{\mathbb{Z}}^{2}}g(n+\frac{1}%
{2},m+\frac{1}{2}) \label{tracedsq}%
\end{equation}
where
\begin{equation}
g(u,v)=2(u^{2}-\frac{1}{4})h(u^{2}\left(  \Lambda a\right)  ^{-2}%
+v^{2}\left(  \Lambda\beta\right)  ^{-2}) \label{gfunctdef}%
\end{equation}
One has
\[
\hat{g}(0,0)=\int_{{\mathbb{R}}^{2}}g(u,v)dudv=2\int_{{\mathbb{R}}^{2}}%
(u^{2}-\frac{1}{4})h(u^{2}\left(  \Lambda a\right)  ^{-2}+v^{2}\left(
\Lambda\beta\right)  ^{-2})dudv
\]%
\[
=2\left(  \Lambda a\right)  \left(  \Lambda\beta\right)  \int_{{\mathbb{R}%
}^{2}}(\left(  \Lambda a\right)  ^{2}x^{2}-\frac{1}{4})h(x^{2}+y^{2})dxdy
\]
using }$u=x\left(  \Lambda a\right)  $ and $v=y\left(  \Lambda\beta\right)  .$
{Thus we get:
\begin{equation}
\hat{g}(0,0)=2\pi\left(  \Lambda\beta\right)  \left(  \Lambda a\right)
^{3}\int_{0}^{\infty}h(\rho^{2})\rho^{3}d\rho-\pi\left(  \Lambda\beta\right)
\left(  \Lambda a\right)  \int_{0}^{\infty}h(\rho^{2})\rho d\rho
\label{ghatzero}%
\end{equation}
To estimate the remainder, given by the sum
$$
\sum_{(n,m)\neq(0,0)}(-1)^{n+m}\hat{g}(n,m)
$$
 we treat separately the Fourier transforms of
\[
g_{1}(u,v)=u^{2}h(u^{2}\left(  \Lambda a\right)  ^{-2}+v^{2}\left(
\beta\Lambda\right)  ^{-2})\,,\ \ g_{2}(u,v)=h(u^{2}\left(  \Lambda a\right)
^{-2}+v^{2}\left(  \Lambda\beta\right)  ^{-2})
\]
One has
\[
\hat{g}_{2}(n,m)=\int_{{\mathbb{R}}^{2}}g_2(u,v)e^{-2\pi i(nu+mv)}dudv
\]%
\[
=\Lambda^2\beta a\int_{{\mathbb{R}%
}^{2}}h(x^{2}+y^{2})e^{-2\pi i(n\Lambda ax+m\Lambda\beta y)}dxdy=\Lambda^2\beta a
\kappa_2(n\Lambda a,m\Lambda\beta)
\]
where the function of two variables $\kappa_2(u,v)$ is the Fourier transform,
\begin{equation}\label{kappadef}
\kappa_2(u,v)=\int_{{\mathbb{R}}^{2}}h(x^{2}+y^{2})e^{-2\pi i(ux+vy)}dxdy=\kappa(u^2+v^2)
\end{equation}
The function $\kappa$ is related to the function $k(x)$  defined by \eqref{transg}, and one has
\begin{equation}\label{expect}
    \kappa(u^2)=\int_{{\mathbb{R}}}k(x^{2})e^{-2\pi iux}dx
\end{equation}
so that $\kappa(u^2)$ is the Fourier transform of $k(x^2)$.

For $g_1$ one has, similarly,
\[
\hat{g}_{1}(n,m)=\int_{{\mathbb{R}}^{2}}g_{1}(u,v)e^{-2\pi i(nu+mv)}dudv
\]%
\[
= \Lambda^4\beta a^3\int_{{\mathbb{R}%
}^{2}}x^{2}h(x^{2}+y^{2})e^{-2\pi i(n\Lambda ax+m\Lambda\beta y)}dxdy= \Lambda^4\beta a^3\kappa_1(n\Lambda a,m\Lambda\beta)
\]
where the function of two variables $\kappa_1(u,v)$ is the Fourier transform,
\[
\kappa_1(u,v)=\int_{{\mathbb{R}}^{2}}x^{2}h(x^{2}+y^{2})e^{-2\pi i(ux+vy)}dxdy
\]
which is given in terms of \eqref{kappadef} by
\begin{equation}\label{kappa1}
    \kappa_1(u,v)=-\pi^{-2}(u^2\kappa''(u^2+v^2)+\frac 12\kappa'(u^2+v^2))
\end{equation}

Now  for  any test
function $h$ in the Schwartz space ${\mathcal{S}
}({\mathbb{R}})$, the function $x^{2}h(x^{2}+y^{2})$ is in the Schwartz space ${\mathcal{S}%
}({\mathbb{R}}^{2})$ and thus we have for its Fourier transform, and any $k>0$, an estimate of
the form
\begin{equation}
|\kappa_1(u,v)|\leq C_{k}(u^{2}+v^{2})^{-k} \label{estifourierh}%
\end{equation}

We thus get, for $k>2$,
\[
|\sum_{(n,m)\neq(0,0)}(-1)^{n+m}\hat{g}_{1}(n,m)|\leq\sum_{(n,m)\neq
(0,0)}|\hat{g}_{1}(n,m)|
\]%
\[
=  \Lambda^4\beta a^3  \sum
_{(n,m)\neq(0,0)}|\kappa_1(n\Lambda a,m\Lambda\beta)|\leq C_{k}\Lambda^4\beta a^3
\sum_{(n,m)\neq(0,0)}((n\Lambda
a)^{2}+(m\Lambda\beta)^{2})^{-k}%
\]%
\[
\leq C_{k}\Lambda^4\beta a^3
(\Lambda\mu)^{-2k}\sum_{(n,m)\neq(0,0)}(n^{2}+m^{2})^{-k}\,,\ \ \mu=\inf(
a,\beta)
\]

We thus get, using a similar estimate for $\hat g_2$,
\begin{theorem}
\label{theorem3}Consider the product geometry $S_a^{3}\times S_{\beta}^{1}$.
  Then one has, for  any test
function $h$ in the Schwartz space ${\mathcal{S}
}({\mathbb{R}})$, the
equality
\begin{equation}
\mathrm{Tr}(h(D^{2}/\Lambda^{2}))=2\pi\Lambda^4\beta a^3\int_{0}^{\infty}h(\rho^{2})\rho^{3}d\rho-\pi\Lambda^2\beta a \int_{0}^{\infty}h(\rho
^{2})\rho d\rho+\epsilon(\Lambda)\label{estispecact}%
\end{equation}
where $\epsilon(\Lambda)=O(\Lambda^{-k})$ for any $k$ is majorized by
$$
|\epsilon(\Lambda)|\leq 2\Lambda^4\beta a^3  \sum
_{(n,m)\neq(0,0)}|\kappa_1(n\Lambda a,m\Lambda\beta)|+ \frac 12\Lambda^2\beta a
 \sum
_{(n,m)\neq(0,0)}|\kappa_2(n\Lambda a,m\Lambda\beta)|\,.\label{estispecactb1}%
$$
with $\kappa_j$ defined in \eqref{kappadef} and \eqref{kappa1}.
\end{theorem}

This implies that all the Seeley coefficients $a_{2n}$ vanish for $n\geq 2$, and we shall
check this directly for $a_4$ and $a_6$ in \S \ref{seeley}.

This vanishing of the Seeley coefficients does not hold for the $4$ sphere and it is worth understanding why one cannot
expect to use the Poisson summation in the same way for the $4$ sphere.
The problem when one tries to use the Poisson formula as above is that,
{\textit{e.g.\/}\ } for the heat kernel, one is dealing with a function like
$|x|e^{-tx^{2}}$ which is not smooth and whose Fourier transform does not have
rapid decay at $\infty$.

\subsection{Specific test functions}
We shall now concretely evaluate the remainder in Theorem \ref{theorem3} for analytic test functions of the form
\begin{equation}\label{han}
    h(x)=P(\pi x)e^{-\pi x}
\end{equation}
where $P$ is a polynomial of degree $d$. The Fourier transforms of the functions of two variables
$x^2h(x^2+y^2)$ and $h(x^2+y^2)$ are of the form
$$\kappa_j(u,v)=P_j(u,v)e^{-\pi (u^2+v^2)}$$
where the $P_j$ are polynomials. More precisely, since the Fourier transform of $e^{-\lambda\pi (x^2+y^2)}$ is
$\frac 1\lambda e^{-\pi \frac{(u^2+v^2)}{\lambda}}$  one obtains the formula for $P_2$ by differentiation at
$\lambda=1$ and get
$$
\kappa_2(u,v)=P( -\partial_\lambda )_{\lambda=1}\frac 1\lambda e^{-\pi \frac{(u^2+v^2)}{\lambda}}
$$
which is of the form
$$
\kappa_2(u,v)=Q(\pi(u^2+v^2)) e^{-\pi  (u^2+v^2) }
$$
where $Q$ is a polynomial of degree $d$. The transformation $P\mapsto Q=T(P)$ is given by
\begin{equation}\label{transT}
    Q(z)=P( -\partial_\lambda )_{\lambda=1}\frac 1\lambda e^{-\frac z\lambda}
\end{equation}

Moreover one then gets
$$
\kappa_1(u,v)=-(2\pi)^{-2}\partial_u^2\kappa_2(u,v)$$ $$=(u^2Z_1(\pi(u^2+v^2)) +Z_2(\pi(u^2+v^2)))e^{-\pi  (u^2+v^2) }
$$
where
\begin{equation}\label{zz}
Z_1=-Q+ 2 Q'- Q''\,, \ \ Z_2=\frac{1}{2\pi}(Q-  Q')
\end{equation}

We let $C_P$ be the sum of the absolute values of the coefficients of $Q=T(P)$.
\medskip

\begin{corollary}
\label{cor1}Consider the product geometry $S_a^{3}\times S_{\beta}^{1}$.
Let $\mu=\inf (a,\beta)$. Then one has, with $h$ any test
function of the form \eqref{han}, the
equality
\begin{equation}
\mathrm{Tr}(h(D^{2}/\Lambda^{2}))=2\pi\Lambda^4\beta a^3\int_{0}^{\infty}h(\rho^{2})\rho^{3}d\rho-\pi\Lambda^2\beta a \int_{0}^{\infty}h(\rho
^{2})\rho d\rho+\epsilon(\Lambda)\label{estispecact1b}%
\end{equation}
where, assuming $\mu\Lambda\geq \sqrt{d(1+\log d)}$ and $\mu\Lambda\geq 1$,
\begin{equation}
|\epsilon(\Lambda)|\leq C e^{-\frac \pi 2 (\mu\Lambda)^2}\,,\ C=\Lambda^4\beta a^3C_P (8+6d+2d^2)
\label{estispecact2b}%
\end{equation}
\end{corollary}

\proof
One has
$$
x^ke^{-x/2}\leq 1 \,, \ \forall x\geq 3 k(1+\log k)
$$
Thus, for $(n,m)\neq (0,0)$ one has
$$
|\kappa_2(n\Lambda a,m\Lambda \beta)|\leq C_P e^{-\frac \pi 2((n\Lambda a)^2+(m\Lambda \beta)^2)}
$$
since $\pi((n\Lambda a)^2+(m\Lambda \beta)^2)\geq 3d(1+\log d)$. Moreover,
since $e^{-\frac \pi 2 (\mu\Lambda)^2}\leq \frac 14$, one gets
$$
\sum_{(n,m)\neq (0,0)}e^{-\frac \pi 2((n\Lambda a)^2+(m\Lambda \beta)^2)}\leq 8 e^{-\frac \pi 2 (\mu\Lambda)^2}
$$
and
$$
\sum
_{(n,m)\neq(0,0)}|\kappa_2(n\Lambda a,m\Lambda\beta)|\leq 8\,C_P e^{-\frac \pi 2 (\mu\Lambda)^2}
$$
A similar estimate using \eqref{zz} yields
$$
\sum
_{(n,m)\neq(0,0)}|\kappa_1(n\Lambda a,m\Lambda\beta)|\leq (2+3d+d^2)\,C_P e^{-\frac \pi 2 (\mu\Lambda)^2}
$$
Thus by Theorem \ref{theorem3}, the inequality \eqref{estispecact2b} holds for $$C=C_P(2\Lambda^4\beta a^3(2+3d+d^2)+4\Lambda^2\beta a)\,.$$ One then uses the hypothesis $\mu\Lambda\geq 1$ to simplify $C$.
\endproof

The meaning of Corollary \ref{cor1} is that the accuracy of the asymptotic expansion is at least of the
order of $e^{-\frac \pi 2 (\mu\Lambda)^2}$. Indeed the term $\Lambda^4\beta a^3$ in \eqref{estispecact2b} is the dominant volume term in the spectral action and the other terms in the formula for $C$ are of order one. Thus for instance for a size $\mu\Lambda\sim 100$ one gets that the asymptotic expansion accurately delivers the first $6820$ decimal places of the spectral action.
Note that some test functions of the form \eqref{han} give excellent approximations to cutoff functions, in particular
\begin{equation}\label{hnfunct}
h_n(x)=\sum_0^n\frac{(\pi x)^k}{k!}\,e^{-\pi x}
\end{equation}
The graph of $h_n(x^2)$ is shown in Figure \ref{figure1} for $n=20$. For $h=h_{20}$ the computation gives $C_P(8+6d+2d^2)\leq 2\times 10^6$ so that this constant only interferes
with the last six decimal places in the above accuracy.

\medskip

In our simplified physical model we test the approximation of the spectral action by its asymptotic expansion for the Euclidean model $$E(t)=S^3_{a(t)}\times S^1_{\beta(t)}$$ where space at a given time $t$ is given  by a sphere with
radius $a(t)$ and $\beta(t)$ is a uniform value of inverse temperature.
One can then easily see that the above  approximation to the spectral action is
fantastically accurate, going backwards in time  all the way   up to one order lower
than the Planck energy. In doing so the radius $a(t)$ varies between at least $\sim 10^{61\text{ }}$ Planck units and  10 Planck units (i.e.
$10^{-34}$ m), while the temperature varies
between   $2.7^{\circ}K$ and $\left(  10^{31}\right)  ^{\circ}K$. It is for an inner size less than
 $10$ in Planck units  that the approximation does break down.

\begin{remark}\label{esticp} {\rm For later purpose it is important to estimate the constant $C_P$ in terms of the coefficients of the polynomial $P$. Let then $P(z)=z^n$. One has $h(x)=(\pi x)^ne^{-\pi x}$ and the function $k(x)$ associated to $h$ by \eqref{transg} is
$$
k(x)=\int_{\mathbb R} h(x+y^2)dy=\pi^ne^{-\pi x}\sum_0^n {n\choose k}x^{n-k}\int_{\mathbb R}y^{2k}e^{-\pi y^2}dy
$$
$$
=\pi^{-1/2}e^{-\pi x}\sum_0^n {n\choose k}\Gamma(\frac 12 +k)(\pi x)^{n-k}
$$
To obtain $Q=T(P)$ one then needs to compute the Fourier transform $\kappa(u^2)$ of the function $k(x^2)$ as in \eqref{expect}. The Fourier transform  of $(\pi x^2)^m e^{-\pi x^2}$ is
$$\ell_m(u)=(-4\pi)^{-m}\partial_u^{2m}e^{-\pi u^2}=L_m(\pi u^2)e^{-\pi u^2}$$
and one checks, using the relation
$$
L_{m+1}(z)=1/2 ((1 - 2 z) L_m(z) + (-1 + 4 z)
L'_m(z) - 2 z L''_m(z)
$$
that the sign of the coefficient of $z^k$ in $L_m(z)$ is $(-1)^k$.
Thus the sum of the absolute values of the coefficients of $L_m$ is equal to $L_m(-1)=\ell_m(i\pi^{-1/2})/e$. Thus since the above sum giving $k(x)$ has positive coefficients we get that, for $P(z)=z^n$, the constant $C_P$ is given by
$Q(\pi(u^2+v^2)e^{-\pi(u^2+v^2)})/e$ for $(u,v)=(i\pi^{-1/2},0)$, which gives
$$
C_P=\int_{{\mathbb R}^2}\pi^n(y^2+x^2)^ne^{-\pi y^2-\pi x^2+2\sqrt \pi x-1}dxdy\,.
$$
One then gets
\begin{equation}\label{esticoefp}
    C_P\leq 2\int_0^\infty u^{2n+1}e^{-(u-1)^2}du=O(\lambda^n n!) \,, \ \forall \lambda>1\,.
\end{equation}
Thus, for an arbitrary polynomial $P(z)=\sum_0^d a_k z^k$ one has
\begin{equation}\label{esticoefp1}
C_P\leq 2\int_0^\infty |P|(u^2) e^{-(u-1)^2}udu\,, \ \ |P|(z)=\sum |a_k|z^k
\end{equation}
}
\end{remark}

\medskip

\subsection{The Higgs potential}

We now look at what happens if one performs the following replacement on the
operator
\[
D^{2}\mapsto D^{2}+H^{2}%
\]
where $H$ is a constant. This amounts with the above notations to the
replacement
\begin{equation}\label{replace}
h(u)\mapsto h(u+H^{2}/\Lambda^{2})\,.
\end{equation}
As long as $H^{2}/\Lambda^{2}$ is of order one, we can trust the asymptotic expansion
and we just need to understand the effect of this shift on the two terms of \eqref{estispecact}.
We look at  the first contribution,
{\textit{i.e.\/}\ }
\[
2\pi\Lambda^4\beta a^3\int
_{0}^{\infty}h(\rho^{2})\rho^{3}d\rho=\pi\Lambda^4\beta a^3\int_{0}^{\infty}uh(u)du
\]
We let $x=H^{2}/\Lambda^{2}$, and get, after the above replacement,
\[
\int_{0}^{\infty}uh(u+x)du=\int_{x}^{\infty}(v-x)h(v)dv=\int_{0}^{\infty
}(v-x)h(v)dv-\int_{0}^{x}(v-x)h(v)dv
\]%
\[
=\int_{0}^{\infty}vh(v)dv-x\int_{0}^{\infty}h(v)dv-\int_{0}^{x}(v-x)h(v)dv
\]
The first term corresponds to the initial contribution of $\pi\Lambda^4\beta a^3\int_{0}^{\infty}uh(u)du$.
The second term gives
\begin{equation}\label{secterm}
-\pi\Lambda^4\beta a^3x\int
_{0}^{\infty}h(v)dv=-\pi\Lambda^2\beta a^3 H^{2}\int_{0}^{\infty}h(v)dv
\end{equation}
which is the expected Higgs mass term from the Seeley--de Witt coefficient $a_{2}$. To understand the last term we assume that $h$ is a cutoff function.

\begin{lemma}
\label{pert}If $\ h$ is a smooth function  constant on the interval
$[0,c]$,  then for  $x=H^{2}/\Lambda^{2}\leq c$ the new terms arising from the replacement
\eqref{replace} are given by
\begin{equation}\label{newterms}
-\pi\Lambda^2\beta a^3  \int_{0}^{\infty
}h(v)dv\;H^{2}+\frac{1}{2}\pi\beta ah(0)\;H^{2}+\frac{1}{2}\pi\beta a^3  h(0)\;H^{4}
\end{equation}
\end{lemma}

\begin{proof}
For the perturbation of $\pi\Lambda^4\beta a^3\int_{0}^{\infty}uh(u)du$, besides
 \eqref{secterm}, we just need to compute the last term $-\int_{0}^{x}(v-x)h(v)dv$, and one has
\[
-\int_{0}^{x}(v-x)h(v)dv=h(0)\int_{0}^{x}(x-v)dv=\frac{1}{2}h(0)x^{2}%
\]
since $h$ is constant on the interval $[0,x]$.

 We then look at the effect on the second contribution,
{\textit{i.e.\/}\ }
\[
-\pi\Lambda^2\beta a  \int_{0}^{\infty
}h(\rho^{2})\rho d\rho=-\frac{1}{2}\pi\Lambda^2\beta a \int_{0}^{\infty}h(u)du
\]
We let, as above, $x=H^{2}/\Lambda^{2}$, and get
\[
\int_{0}^{\infty}h(u+x)du=\int_{x}^{\infty}h(v)dv=\int_{0}^{\infty}%
h(v)dv-\int_{0}^{x}h(v)dv
\]
Thus the perturbation, under the hypothesis of Lemma \ref{pert} is
\[
-\frac{1}{2}\pi\Lambda^2\beta a
(-xh(0))=\frac{1}{2}\pi\beta ah(0)\;H^{2}%
\]
\end{proof}

The three terms in  formula \eqref{newterms} correspond to the following new terms for the spectral action \begin{itemize}
  \item The Higgs mass term coming from the Seeley--de Witt coefficient $a_{2}$.
  \item The $R H^2$ term coming from the Seeley--de Witt coefficient $a_{4}$.
  \item The Higgs potential term in $H^4$ coming from the Seeley--de Witt coefficient $a_{4}$.
\end{itemize}

We can now state the analogue of Theorem \ref{theorem3} as follows

\begin{theorem}
\label{poissonthm1} Consider the product geometry $S_a^{3}\times S_{\beta}^{1}$.
Let $\mu=\inf (a,\beta)$. Then one has, with $h$ any test
function of the form \eqref{han},  the equality
$$
{\mathrm{Tr}}(h((D^{2}+H^{2})/\Lambda^{2}))=2\pi\Lambda^4\beta a^3\int_{0}^{\infty}h(\rho^{2})\rho^{3}d\rho-\pi\Lambda^2\beta a \int_{0}^{\infty}h(\rho
^{2})\rho d\rho
$$
\[
+\pi\Lambda^4\beta a^3 \,V(H^2/\Lambda^2)+\frac 12 \pi\Lambda^2\beta a \, W(H^2/\Lambda^2)+\epsilon(\Lambda)
\]
where
\begin{equation}\label{pot}
V(x)=\int_{0}^{\infty}u(h(u+x)-h(u))du\,, \ W(x)=\int_{0}^{x} h(u)du\,
\end{equation}
and, assuming $\mu\Lambda\geq \sqrt{d(1+\log d)}$, $\mu\Lambda\geq 1$, and
 $H^{2}\Lambda^{-2}\leq c/\pi$,
\begin{equation}
|\epsilon(\Lambda)|\leq C e^{-\frac \pi 2 (\mu\Lambda)^2}\,,\ C=\Lambda^4\beta a^3C'_P (8+6d+2d^2)
\end{equation}
where, with $P(z)=\sum_0^d a_k z^k$ one has  $$C'_P=4\int_0^\infty |P|(u^2+c) e^{-(u-1)^2}udu\,, \ \ |P|(z)=\sum |a_k|z^k$$
\end{theorem}

\proof The new terms simply express the replacement \eqref{replace} in the formula of Theorem \ref{theorem3}. The new function $\tilde h$ thus obtained is still of the form \eqref{han} since it is obtained from $h$ by a translation. It thus only remains to estimate $C_{\tilde P}$
where $\tilde P$ is the polynomial such that $\tilde h(u)=\tilde P(\pi u)e^{-\pi u}$. For $P(z)=z^n$ the constant $C_{\tilde P}$ for a translation $u\mapsto u+x$, $x\geq 0$ of the variable, is less than the constant $C_{P_x}$ for the polynomial $$P_x(\pi u)=P(\pi(u+x))=\sum {n\choose k}(\pi x)^{n-k}(\pi u)^k$$ Thus, by Remark \ref{esticp}, \eqref{esticoefp1}, the constant $C_{P_x}$ is estimated by
$$
C_{P_x}\leq 2\int_0^\infty (u^2+\pi x)^ne^{-(u-1)^2}udu
$$
which is an increasing function of $x$ and thus only needs to be controlled for $x=c/\pi$ in our case.
\endproof

For instance, for $h=h_{20}$ the computation gives $C'_P(8+6d+2d^2)\leq 3\times 10^7$ for $c=1$, so that this constant only interferes
with the last seven decimal places in the accuracy which is the same as in Corollary \ref{cor1}.

Moreover as shown in Lemma \ref{pert}, when $h$ is close to a true cutoff function
\begin{equation}\label{approchiggs}
\pi\Lambda^4\beta a^3 \,V(H^2/\Lambda^2)+\frac 12 \pi\Lambda^2\beta a \, W(H^2/\Lambda^2)
\end{equation}
\[
=-2\pi\Lambda^2\beta a^{3}\int_{0}^{\infty}h(\rho
^{2})\rho d\rho\, H^{2}+ \frac{1}{2}\pi\beta ah(0)\;H^{2}
+\frac{1}{2}\pi\beta a^{3}h(0)\;H^{4}+\delta
\]
where the remainder $\delta$ is estimated from the Taylor expansion of $h$ at $0$. For instance for the functions $h_n$ of \eqref{hnfunct}, one has by construction $0\leq h_n(x)\leq 1$ for all $x$ and
since $$h_n(x)=1 - \sum a(n,k) x^{n+k+1}\,, \ \
a(n,k)=(-1)^k /((n + k + 1) n! k!)$$
one gets, for $h=h_n$ the estimate
$$
|\delta|\leq \pi\Lambda^4\beta a^3\frac{x^{n+3}}{(n+3)(n+1)!}+\pi\Lambda^2\beta a
\frac{x^{n+2}}{2(n+2)!}\,, \ \ x=H^2/\Lambda^2\,.
$$

While the function $W$ is by construction the primitive of $h$, and is increasing
for $h\geq 0$ one has, under the hypothesis of positivity of $h$,

\begin{lemma}
\label{pert1} The function $V(x)$ is decreasing with derivative given by
\[
V^{\prime}(x)=-\int_{x}^{\infty}h(v)dv
\]
The second derivative of $V(x)$ is equal to $h(x)$.
\end{lemma}

\begin{proof}
One has
\[
V^{\prime}(x)=\int_{0}^{\infty}uh^{\prime}(u+x)du=[uh(u+x)]_{0}^{\infty}%
-\int_{0}^{\infty}h(u+x)du
\]
which gives the required results.
\end{proof}

\section{Seeley--De Witt coefficients and Spectral Action on $S^{3}\times S^{1}$}\label{seeley}

In this section we shall compute the asymptotic expansion of the spectral action on the background
geometry of $S^{3}\times S^{1}$ using heat kernel methods. This will enable us
to check independently the accuracy of the estimates derived in the last
section. This background is physically relevant since it can be connected with
simple cosmological models. We refer to \cite{Gilkey}, \cite{cc2} for the
formulas and the method of the computation. The general method we use is also
explained in great detail in a forthcoming paper \cite{hidden}. We start by
computing $a_{0}:$
\begin{align*}
a_{0} &  =\frac{\text{Tr}(1)}{16\pi^{2}}%
%TCIMACRO{\dint }%
%BeginExpansion
{\displaystyle\int}
%EndExpansion
\sqrt{g}d^{4}x=\frac{1}{4\pi^{2}}%
%TCIMACRO{\dint \limits_{S^{3}}}%
%BeginExpansion
{\displaystyle\int\limits_{S^{3}}}
%EndExpansion
\sqrt{^{3}g}d^{3}x%
%TCIMACRO{\dint \limits_{S^{1}}}%
%BeginExpansion
{\displaystyle\int\limits_{S^{1}}}
%EndExpansion
dx\\
&  =\frac{1}{4\pi^{2}}\left(  2\pi^{2}a^{3}\right)  \left(  2\pi\beta\right)
=\pi\beta a^{3}%
\end{align*}
where $\beta$  is the radius of $S_\beta^{1}$  and the volume of the
three sphere $S^3_a$ of radius $a$ is $2\pi^{2}a^{3}$ \cite{weinberg} .

Next we calculate $a_{2}$%
\[
a_{2}=\frac{1}{16\pi^{2}}%
%TCIMACRO{\dint }%
%BeginExpansion
{\displaystyle\int}
%EndExpansion
d^{4}x\sqrt{g}\text{Tr}\left(  E+\frac{1}{6}R\right)
\]
where $E$ is defined from the relation
\[
D^{2}=-\left(  g^{\mu\nu}\nabla_{\mu}\nabla_{\nu}+E\right)
\]
where for pure gravity we have
\[
E=-\frac{1}{4}R
\]
so that (using Tr$\left(  1\right)  =4$ )
\[
a_{2}=\frac{1}{4\pi^{2}}\left(  -\frac{R}{12}\right)
%TCIMACRO{\dint }%
%BeginExpansion
{\displaystyle\int}
%EndExpansion
d^{4}x\sqrt{g}%
\]
since the curvature is constant. The curvature tensor\footnote{the sign convention for this tensor is the same as in \cite{Gilkey}} is, using the coordinates of \cite{weinberg} for the three sphere $S^3_a$ with labels $i,j,k,l$ and the label $4$ for the coordinate in $S_\beta^1$,
\begin{align*}
R_{ijkl} &  =-a^{-2}\left(  g_{ik}g_{jl}-g_{il}g_{jk}\right)  ,\quad
i,j,k,l=1,\cdots3\\
R_{ijk4} &  =0\\
R_{i4j4} &  =0
\end{align*}
where $g_{ij}$ is the metric on the three sphere as in \cite{weinberg}. The Ricci tensor is given, following the sign convention of \cite{Lawson}
which introduces a minus sign in passing from the curvature tensor to the Ricci tensor, by
\begin{align*}
R_{ij} &  =-g^{kl}R_{ikjl}=2a^{-2}g_{ij}\\
R_{i4} &  =0\\
R_{44} &  =0
\end{align*}
Thus the scalar curvature is
\[
R=g^{ij}R_{ij}=\frac{6}{a^{2}}%
\]
and the $a_{2}$ term in the heat expansion simplifies to
\[
a_{2}=-\pi\beta a\left(  \frac{1}{2}\right)
\]
Next for $a_{4}$ we have
\begin{align*}
a_{4} &  =\frac{1}{16\pi^{2}}\frac{1}{360}%
%TCIMACRO{\dint \limits_{M}}%
%BeginExpansion
{\displaystyle\int\limits_{M}}
%EndExpansion
d^{4}x\sqrt{g}\,Tr\left(  12R_{;\mu}^{\;\;\mu}+5R^{2}-2R_{\mu\nu}R^{\mu\nu
}\right.  \\
&  \quad\left.  +\,2R_{\mu\nu\rho\sigma}R^{\mu\nu\rho\sigma}+60RE+180E^{2}%
+60E_{;\mu}^{\quad\mu}+30\,\Omega_{\mu\nu}\Omega^{\mu\nu}\right)
\end{align*}
where for the pure gravitational theory, we have
\[
E=-\frac{1}{4}R,\qquad\Omega_{\mu\nu}=\frac{1}{4}R_{\mu\nu}^{\quad ab}%
\gamma_{ab}%
\]
In this case it was shown in \cite{cc2} that $a_{4}$ reduces to
\begin{align}
a_{4} &  =\frac{1}{4\pi^{2}}\frac{1}{360}%
%TCIMACRO{\dint }%
%BeginExpansion
{\displaystyle\int}
%EndExpansion
d^{4}x\sqrt{g}\left(  5R^{2}-8R_{\mu\nu}^{2}-7R_{\mu\nu\rho\sigma}^{2}\right)
\nonumber\\
&  =\frac{1}{4\pi^{2}}\frac{1}{360}%
%TCIMACRO{\dint }%
%BeginExpansion
{\displaystyle\int}
%EndExpansion
d^{4}x\sqrt{g}\left(  -18C_{\mu\nu\rho\sigma}^{2}+11R^{\ast}R^{\ast}\right)
\end{align}
which is obviously scale invariant. The Weyl tensor $C_{\mu\nu\rho\sigma}$ is
defined by
\begin{align*}
C_{\mu\nu\rho\sigma} &  =R_{\mu\nu\rho\sigma}+\frac{1}{2}\left(  R_{\mu\rho
}g_{\nu\sigma}-R_{\nu\rho}g_{\mu\sigma}-R_{\mu\sigma}g_{\nu\rho}+R_{\nu\sigma
}g_{\mu\rho}\right)  \\
&  -\frac{1}{6}\left(  g_{\mu\rho}g_{\nu\sigma}-g_{\nu\rho}g_{\mu\sigma
}\right)  R
\end{align*}
This tensor vanishes on $S^{3}\times S^{1}$ as can be seen by evaluating the
components
\begin{align*}
C_{ijkl} &  =a^{-2}\left[  -(g_{ik}g_{jl}-g_{il}g_{jk})+2(g_{ik}g_{jl}%
-g_{il}g_{jk})-(g_{ik}g_{jl}-g_{il}g_{jk})\right]=0  \\
C_{ijk4} &  =0\\
C_{i4k4} &  =0
\end{align*}
Similarly the Gauss-Bonnet term
\begin{align*}
R^{\ast}R^{\ast} &  =\frac{1}{4}\epsilon^{\mu\nu\rho\sigma}\epsilon
_{\alpha\beta\gamma\delta}R_{\mu\nu}^{\quad\alpha\beta}R_{\rho\sigma}%
^{\quad\gamma\delta}\\
&  =\epsilon^{ijk4}\epsilon_{\alpha\beta\gamma\delta}\left(  R_{ij}%
^{\quad\alpha\beta}R_{k4}^{\quad\gamma\delta}\right)  \\
&  =0
\end{align*}
The next step of calculating $a_{6}$ is in general extremely complicated, but
for spaces of constant curvature the expression simplifies as all covariant
derivatives of the curvature tensor, Riemann tensor and scalar curvature
vanish. The non-vanishing terms are, using Theorem 4.8.16 of \cite{Gilkey} and the
above sign convention for the Ricci tensor $R_{\mu\nu}$ and the scalar curvature,
\begin{align*}
a_{6} &  =\frac{1}{16\pi^{2}}%
%TCIMACRO{\dint }%
%BeginExpansion
{\displaystyle\int}
%EndExpansion
d^{4}x\sqrt{g}\text{Tr}\left(  \frac{1}{9\cdot7!}\left(  35R^{3}-42RR_{\mu\nu
}^{2}+42RR_{\mu\nu\rho\sigma}^{2}\right.  \right.  \\
&  \qquad\qquad\qquad-208R_{\mu\nu}R_{\mu\rho}R_{\nu\rho}-192R_{\mu\rho}%
R_{\nu\sigma}R_{\mu\nu\rho\sigma}-48R_{\mu\nu}R_{\mu\rho\sigma\kappa}%
R_{\nu\rho\sigma\kappa}\\
&  \qquad\qquad\qquad\left.  -44R_{\mu\nu\rho\sigma}R_{\mu\nu\kappa\lambda
}R_{\rho\sigma\kappa\lambda}-80R_{\mu\nu\rho\sigma}R_{\mu\kappa\rho\lambda
}R_{\nu\kappa\sigma\lambda}\right)  \\
&  \qquad\qquad+\frac{1}{360}\left(  -12\Omega_{\mu\nu}\Omega_{\nu\rho}%
\Omega_{\rho\mu}-6R_{\mu\nu\rho\sigma}\Omega_{\mu\nu}\Omega_{\rho\sigma
}-4R_{\mu\nu}\Omega_{\mu\rho}\Omega_{\nu\rho}+5R\Omega_{\mu\nu}^{2}\right.  \\
&  \qquad\qquad\qquad\qquad\left.  \left.  +60E^{3}+30E\Omega_{\mu\nu}%
^{2}+30RE^{2}+5R^{2}E-2R_{\mu\nu}^{2}E+2R_{\mu\nu\rho\sigma}^{2}E\right)
\right)
\end{align*}
We can now compute each of the above eighteen terms. These are listed in an
appendix. Collecting these terms we obtain that the integrand is
\begin{align*}
&  -\frac{4a^{-6}}{9\cdot7!}\left(  -35\cdot6^{3}+42\cdot72-42\cdot
72+208\cdot24-192\cdot24+48\cdot24-44\cdot24-80\cdot6\right) \\
&  -\frac{4a^{-6}}{360}\left(  9+18-12+45+\frac{15\cdot27}{2}-\frac{5\cdot
27}{2}-15\cdot27+10\cdot27-36+36\right) \\
&  =a^{-6}\left(  \frac{2}{3}-\frac{2}{3}\right)  =0
\end{align*}
implying that
\[
a_{6}=0\,,
\]
which shows that the cancelation is highly non-trivial. We conclude that the
spectral action, up to terms of order $\frac{1}{\Lambda^{4}}$ is given by
\[
S=\Lambda^{4}%
{\displaystyle\int_{0}^{\infty}}
xh\left(  x\right)  dx\left(  \pi\beta a^{3}\right)  -\Lambda^{2}%
{\displaystyle\int_{0}^{\infty}}
h\left(  x\right)\,dx  \left(  \pi\beta a\frac{1}{2}\right)  +O\left(
\Lambda^{-4}\right)
\]
After making the change of variables $x=\rho^{2}$ we get
\[
S=\left(  \pi\beta\Lambda\right)  \left[  2\left(  \Lambda a\right)  ^{3}%
{\displaystyle\int_{0}^{\infty}}
\rho^{3}h\left(  \rho^{2}\right)  d\rho-\left(  \Lambda a\right)
{\displaystyle\int_{0}^{\infty}}
\rho h\left(  \rho^{2}\right)  d\rho\right]  +O\left(  \Lambda^{-4}\right)
\]
This confirms equation (\ref{estispecact}) and shows that, to a very high
degree of accuracy, the spectral action on $S^{3}\times S^{1}$ is given by the
first two terms.

\begin{remark}\label{11rem} {\rm
It is worth noting that one can also check the value of the Gauss-Bonnet term
on $S^{4}$ and show that it agrees with the value obtained in \eqref{4sphere}.
To see this note that the Riemann tensor in this case is given by (\cite{weinberg})
\[
R_{\mu\nu\rho\sigma}=-a^{-2}\left(  g_{\mu\rho}g_{\nu\sigma}-g_{\mu\sigma
}g_{\nu\rho}\right)
\]
which implies\footnote{One can double check the value of $R^{\ast}R^{\ast}$ using the
Gauss--Bonnet Theorem.  } that
\begin{align*}
C_{\mu\nu\rho\sigma} &  =0\\
R^{\ast}R^{\ast} &  =6 a^{-4}
\end{align*}
and thus
\[
a_{4}=\frac{1}{4\pi^{2}}\frac{11}{60}%
%TCIMACRO{\dint _{S^{4}}}%
%BeginExpansion
a^{-4}{\displaystyle\int_{S^{4}}}
%EndExpansion
d^{4}x\sqrt{g}%
\]
The volume of $S^{4}$ is
\[
V_{4}=%
%TCIMACRO{\dint _{S^{4}}}%
%BeginExpansion
{\displaystyle\int_{S^{4}}}
%EndExpansion
d^{4}x\sqrt{g}=\frac{2\pi^{\frac{5}{2}}}{\Gamma\left(  \frac{5}{2}\right)
}a^{4}=\frac{8\pi^{2}}{3}a^{4}%
\]
and this implies that
\[
a_{4}=\frac{11}{360}\times 4%
\]
which agrees exactly with the calculation of  \eqref{4sphere} based on zeta functions.}
\end{remark}

\medskip

\begin{center}
{\Large \textbf{Appendix}}

\bigskip
\end{center}

In this appendix we compute the eighteen non-vanishing terms that appear in
the $a_{6}$ term of the heat kernel expansion. \ Using the properties
\begin{align*}
R_{\mu\nu}^{2} &  =R_{ij}^{2}=12a^{-4}\\
R_{\mu\nu\rho\sigma}^{2} &  =R_{ijkl}^{2}=12a^{-4}%
\end{align*}

\begin{align*}
35R^{3}  &  =35(6)^{3}a^{-6}\\
-42RR_{\mu\nu}^{2}  &  =-42\left(  6\right)  \left(  12\right)  a^{-6}\\
-208R_{\mu\nu}R_{\mu\rho}R_{\nu\rho}  &  =-208\left(  2\right)  ^{3}%
g_{ij}g_{ik}g_{jk}=208\left(  2\right)  ^{3}\left(  3\right)  a^{-6}\\
-192R_{\mu\rho}R_{\nu\sigma}R_{\mu\nu\rho\sigma}  &  =-192R_{ik}R_{jl}%
R_{ijkl}\\
&  =192\left(  2\right)  ^{2}g_{ik}g_{jl}\left(  g_{ik}g_{jl}-g_{il}%
g_{jk}\right)  a^{-6}\\
&  =192\left(  24\right)  a^{-6}\\
-48R_{\mu\nu}R_{\mu\rho\sigma\kappa}R_{\nu\rho\sigma\kappa}  &  =-48\left(
2\right)  g_{ij}\left(  g_{ik}g_{lm}-g_{il}g_{km}\right)  \left(  g_{jk}%
g_{lm}-g_{jl}g_{km}\right)  a^{-6}\\
&  =-48\left(  4\right)  g_{ij}\left(  2g_{ij}\right)  a^{-6}=-48\left(
24\right)  a^{-6}\\
-44R_{\mu\nu\rho\sigma}R_{\mu\nu\kappa\lambda}R_{\rho\sigma\kappa\lambda}  &
=44\left(  g_{ik}g_{jl}-g_{il}g_{jk}\right)  \left(  g_{ip}g_{jq}-g_{iq}%
g_{jp}\right)  \left(  g_{kp}g_{lq}-g_{lq}g_{lp}\right)  a^{-6}\\
&  =44\left(  4\right)  \left(  6\right)  a^{-6}\\
-80R_{\mu\nu\rho\sigma}R_{\mu\kappa\rho\lambda}R_{\nu\kappa\sigma\lambda}  &
=80\left(  g_{ik}g_{jl}-g_{il}g_{jk}\right)  \left(  g_{ik}g_{pq}-g_{iq}%
g_{pk}\right)  \left(  g_{jl}g_{pq}-g_{jq}g_{pl}\right)  a^{-6}\\
&  =80\left(  3g_{jl}g_{pq}-g_{pq}g_{jl}-g_{lj}g_{pq}+g_{lq}g_{jp}\right)
\left(  g_{jl}g_{pq}-g_{jq}g_{pl}\right)  a^{-6}\\
&  =80\left(  9-3+3-3\right)  a^{-6}\\
&  =80\left(  6\right)  a^{-6}%
\end{align*}
Collecting the first set of terms we get
\begin{align*}
&  -\frac{4a^{-6}}{9\cdot7!}\left(  -35\cdot6^{3}+42\cdot72-42\cdot
72+208\cdot24-192\cdot24+48\cdot24-44\cdot24-80\cdot6\right) \\
&  =\frac{2}{3}a^{-6}%
\end{align*}
Now we continue with the second set of terms%
\begin{align*}
-12\text{Tr}\left(  \Omega_{\mu\nu}\Omega_{\nu\rho}\Omega_{\rho\mu}\right)
&  =-12\left(  \frac{1}{4}\right)  ^{3}\text{Tr}\left(  \gamma_{ab}\gamma
_{cd}\gamma_{ef}\right)  R_{\mu\nu}^{\quad ab}R_{\nu\rho}^{\quad cd}R_{\rho
\mu}^{\quad ef}\\
&  =\text{Tr}\left(  1\right)  12\left(  \frac{1}{4}\right)  ^{3}\left(
8\right)  R_{\mu\nu ab}R_{\nu\rho bc}R_{\rho\mu ac}\\
&  =-\frac{3}{2}\left(  g_{ik}g_{jl}-g_{il}g_{jk}\right)  \left(  g_{jl}%
g_{pq}-g_{jp}g_{lq}\right)  \left(  g_{pk}g_{iq}-g_{pi}g_{kq}\right)
a^{-6}\text{Tr}\left(  1\right) \\
&  =-\frac{3}{2}\left(  3g_{ik}g_{pq}-g_{ik}g_{pq}-g_{ik}g_{pq}+g_{iq}%
g_{pk}\right)  \left(  g_{pk}g_{iq}-g_{pi}g_{kq}\right)  4a^{-6}\\
&  =-\frac{3}{2}\left(  3-3+9-3\right)  4a^{-6}\\
&  =-9\cdot 4a^{-6}%
\end{align*}%
\begin{align*}
-6R_{\mu\nu\rho\sigma}\text{Tr}\left(  \Omega_{\mu\nu}\Omega_{\rho\sigma
}\right)   &  =-\frac{6}{4^{2}}R_{\mu\nu\rho\sigma}\text{Tr}\left(
\gamma_{ab}\gamma_{cd}\right)  R_{\mu\nu}^{\quad ab}R_{\rho\sigma}^{\quad
cd}\\
&  =\frac{12}{16}R_{\mu\nu\rho\sigma}R_{\mu\nu}^{\quad ab}R_{\rho\sigma
ab}\text{Tr}\left(  1\right) \\
&  =-\frac{3}{4}\left(  g_{ik}g_{jl}-g_{il}g_{jk}\right)  \left(  g_{ip}%
g_{jq}-g_{iq}g_{jp}\right)  \left(  g_{kp}g_{lq}-g_{kq}g_{lp}\right)
4a^{-6}\\
&  =-3\left(  9-3\right)  4a^{-6}=-18\cdot4a^{-6}%
\end{align*}%
\begin{align*}
-4R_{\mu\nu}\text{Tr}\left(  \Omega_{\mu\rho}\Omega_{\nu\rho}\right)   &
=-\frac{1}{4}R_{\mu\nu}\text{Tr}\left(  \gamma_{ab}\gamma_{cd}\right)
R_{\mu\rho}^{\quad ab}R_{\nu\rho}^{\quad cd}\\
&  =\frac{1}{2}R_{\mu\nu}R_{\mu\rho}^{\quad ab}R_{\nu\rho ab}^{\quad}%
\text{Tr}\left(  1\right) \\
&  =\frac{1}{2}\left(  2\right)  g_{ij}\left(  g_{ip}g_{mq}-g_{iq}%
g_{mp}\right)  \left(  g_{jp}g_{mq}-g_{jq}g_{mp}\right)  4a^{-6}\\
&  =2\left(  9-3\right)  4a^{-6}=12\cdot4a^{-6}%
\end{align*}%
\begin{align*}
5R\text{Tr}\left(  \Omega_{\mu\nu}^{2}\right)   &  =\frac{5}{16}%
R\text{Tr}\left(  \gamma_{ab}\gamma_{cd}\right)  R_{\mu\nu}^{\quad ab}%
R_{\mu\nu}^{\quad cd}\\
&  =-\frac{5}{8}RR_{\mu\nu\rho\sigma}^{2}\text{Tr}\left(  1\right) \\
&  =-\frac{5}{8}\left(  6\right)  \left(  12\right)  4a^{-6}\\
&  =-45\cdot4a^{-6}%
\end{align*}

\begin{align*}
60\text{Tr}\left(  E^{3}\right)   &  =60\left(  -\frac{1}{4}\right)  ^{3}%
R^{3}\text{Tr}\left(  1\right) \\
&  =60\left(  -\frac{1}{4}\right)  ^{3}\left(  6\right)  ^{3}\cdot4a^{-6}\\
&  =-\frac{1}{2}\left(  15\cdot27\right)  \cdot4a^{-6}%
\end{align*}

\begin{align*}
30\text{Tr}\left(  E\Omega_{\mu\nu}^{2}\right)   &  =30\left(  -\frac{3}%
{2}\right)  (-2)\left(  -\frac{1}{4}\right)  ^{2}a^{-2}R_{\mu\nu\rho\sigma
}^{2}\text{Tr}\left(  1\right) \\
&  =\frac{90}{16}(12)4a^{-6}\\
&  =\frac{1}{2}\left(  5\cdot27\right)  \cdot4a^{-6}%
\end{align*}%
\begin{align*}
30RE^{2}\text{Tr}\left(  1\right)   &  =\frac{30}{16}R^{3}\text{Tr}\left(
1\right) \\
&  =\left(  \frac{15}{8}\right)  \left(  6\right)  ^{3}4a^{-6}\\
&  =\left(  15\cdot27\right)  \cdot4a^{-6}%
\end{align*}%
\begin{align*}
5R^{2}E\text{Tr}\left(  1\right)   &  =-\frac{5}{4}R^{3}\cdot4\\
&  =-\frac{5}{4}\left(  6\right)  ^{3}4a^{-6}\\
&  =-\left(  10\cdot27\right)  \cdot4a^{-6}%
\end{align*}%
\begin{align*}
-2R_{\mu\nu}^{2}E\text{Tr}\left(  1\right)   &  =-2R_{\mu\nu}^{2}\left(
-\frac{R}{4}\right)  4\\
&  =-2\left(  12\right)  \left(  -\frac{3}{2}\right)  4a^{-6}\\
&  =36\cdot4a^{-6}%
\end{align*}%
\begin{align*}
2R_{\mu\nu\rho\sigma}^{2}E\text{Tr}\left(  1\right)   &  =2\left(  12\right)
\left(  -\frac{3}{2}\right)  4a^{-6}\\
&  =-36\cdot4a^{-6}%
\end{align*}
Collecting the second set of terms we get
\begin{align*}
-  &  \frac{4a^{-6}}{360}\left(  9+18-12+45+\frac{15\cdot27}{2}-\frac
{5\cdot27}{2}-15\cdot27+10\cdot27-36+36\right) \\
&  =-\frac{2}{3}a^{-6}%
\end{align*}
Thus the sum of all the terms in $a_{6}$ is zero.

\vspace{0.3cm} {\centerline{\large\bf {Acknowledgements}}} \vspace{0.3cm}

The research of A. H. C. is supported in part by the Arab Fund for Social and
Economic Development.


\begin{thebibliography}{99}                                                                                               %
\bibitem {Milnor}J.~Milnor, \emph{Eigenvalues of the Laplace operator on
certain manifolds} Proc. Natl. Acad. Sci. U S A. 51(4) (1964), 542.

\bibitem {CoRec}A.~Connes, \emph{On the spectral characterization of
manifolds}, to appear.

\bibitem {CoSM}A.~Connes, \emph{Gravity coupled with matter and the foundation
of noncommutative geometry}, Comm. Math. Phys. \textbf{182 }(1996) 155-176.

\bibitem {Barrett}John Barrett, \emph{The Lorentzian Version of the
Noncommutative Geometry Model of Particle Physics}, \textit{J. Math. Phys.
}\textbf{48: }012303 (2007).

\bibitem {AC}A. Connes, \emph{Noncommutative geometry and the standard model
with neutrino mixing,} JHEP \textbf{0611:081 }(2006)

\bibitem {beggar}A.~Chamseddine and A.~Connes, \emph{Why the Standard Model},
Jour. Geom. Phys. \textbf{58 }(2008) 38-47.

\bibitem {SM}A. Chamseddine and A. Connes, \emph{Conceptual explanation for
the algebra in the noncommutative approach to the standard model}, Phys. Rev.
Lett. \textbf{99 }(2007) 191601.

\bibitem {cc2}A.~Chamseddine and A. Connes, \emph{The Spectral action
principle}, Comm. Math. Phys. \textbf{186} (1997), 731--750.

\bibitem {mc2}A.~Chamseddine, A.~Connes, M.~Marcolli, \emph{Gravity and the
standard model with neutrino mixing}, Adv. Theor. Math. \textbf{11} (2007) 991-1090.

\bibitem {dilaton}A.~Chamseddine and A.~Connes, \emph{Scale invariance in the
spectral action, }Jour. Math. Phys. \textbf{47} (2006) 063504.

\bibitem {boundary}A.~Chamseddine and A.~Connes, \emph{Quantum gravity
boundary terms from the spectral action of \ noncommutative space, }Phys. Rev.
Lett. \textbf{99} (2007) 071302.

\bibitem {GH}G. Gibbons and S. Hawking, \emph{Action integrals and partition
functions in quantum gravity, }Phys. Rev. D \textbf{15} (1977) 2752.

\bibitem {ashtekar}A. Ashtekar, J. Engle and D. Sloan, \emph{Asymptotics and
Hamiltonians in a first order formalism, }Class. Quant. Grav. \textbf{25}
(2008) 095020.

\bibitem {uni}J. van der Bij, H. van Dam and Y. Ng, Physica A\textbf{116
}(1982) 307; F. Wilczek and A. Zee in High Energy Physics, ed. S. Mintz and
A.Perlmutter, Plenum, NY, 1985; S. Weinberg, Rev. Mod. Phys. \textbf{61
}(1989) 1.

\bibitem {reuter}M. Reuter, \emph{Nonperturbative evolution equation for
quantum gravity, }Phys. Rev. D \textbf{57 }(1998) 971.

\bibitem {Dou}D.~Dou, R.~Percacci, \emph{The running gravitational couplings,
}Class. Quant. Grav. \textbf{15} (1998) 3449.

\bibitem {Perc}R.~Percacci, \emph{Renormalization group, systems of units and
the hierarchy problem}, J. Phys. \textbf{A40} (2007) 4895.

\bibitem {hidden}A. Chamseddine and A. Connes, \emph{Unocovering the
noncommutative geometry of space-time: a user manual for physicists, }to appear.

\bibitem {Hambey}T. Hambye and K. Riesselmann, \emph{Matching Conditions and
Upper bounds for Higgs Masses Revisited}, \textit{Phys. Rev. }\textbf{D55:
}(1997) 7255.

\bibitem {sturmia}G. Isidori, V. Rychkov, A. Sturmia and N. Tetradis,
\emph{Gravitational corrections to the standard model vacuum decay, }Phys.
Rev. D\textbf{77 }(2008) 025034.

\bibitem {CW}S. Coleman and E. Weinberg, \emph{Radiative corrections as the
origin of spontaneous symmetry breaking}, \textit{Phys. Rev. }\textbf{D7,
}(1973) 1888.

\bibitem {Bubu}W. Buchm\"{u}ller and C. Busch, \emph{Symmetry breaking and
mass bounds in the standard model with hidden scale invariance}, \textit{Nucl.
Phys. }\textbf{B 349 }(1991) 71.

\bibitem {Gilkey}P. Gilkey, \emph{Invariance Theory, the Heat Equation and the
Atiyah-Singer Index Theorem}\textit{, } Wilmington, Publish or Perish, 1984.

\bibitem{Lawson} H. B.~Lawson, M-L.~Michelsohn {\em Spin geometry},
Princeton Mathematical Series, 38. Princeton University Press,
Princeton, NJ, 1989.

\bibitem {weinberg}S. Weinberg, \emph{Gravitation and Cosmology, }J. Wiley,
(1972) pages 389--390.

\bibitem {trautman}A. Trautman, \emph{Spin structures on hypersurfaces and
spectrum of Dirac ooperators on spheres, }in Spinors, Twistors, Clifford
Algebras and  Quantum Deformations, Kluver Academic Publishers 1993.

\bibitem {Higuchi}R. Camporesi and A. Higuchi, \emph{On the eigenfunctions of
Dirac operators on spheres and hyperbolic spaces, }J. Geom. Phys. \textbf{20
}(1996) 1.\emph{ }
\end{thebibliography}
\end{document}